\documentclass[12pt]{article}

\usepackage{minted}
\usepackage{natbib}
\usepackage[margin=25mm]{geometry}
\usepackage{booktabs}
\usepackage{siunitx}
\usepackage{hyperref}
\usepackage{threeparttable}

\usepackage{ascmac,ulem,xcolor,eqnarray,tikz,mathcomp,setspace,amsmath,url,subcaption,bm,physics}
\usepackage{amsthm,amssymb}
\usepackage{enumerate}
\usepackage{comment,here}
\usepackage{changepage}
\usepackage{pdflscape}
\theoremstyle{definition}
\newtheorem{theorem}{Theorem}

\newtheorem{prop}{Proposition}
\newtheorem{assumption}{Assumption}

\def\modelname{SAR }

\usepackage{graphicx}
\usepackage{multirow}
\usepackage{arydshln}
\hypersetup{%
setpagesize=false,%
bookmarks=true,%
bookmarksdepth=tocdepth,%
bookmarksnumbered=true,%
colorlinks=true,%
linkcolor=blue,
citecolor=blue,
pdftitle={},%
pdfsubject={},%
pdfauthor={},%
pdfkeywords={}}

\allowdisplaybreaks

\newcommand{\argmin}{\mathop{\rm arg~min}\limits}

\begin{document}

\bibliographystyle{econ-aea}

\title{Identification and Bayesian Inference for Synthetic Control Methods with Spillover Effects}
\author{Shosei Sakaguchi\thanks{Faculty of Economics, The University of Tokyo, 7-3-1 Hongo, Bunkyo-ku, Tokyo 113-0033, Japan. Email:$\ $sakaguchi@e.u-tokyo.ac.jp. } \ and\  Hayato Tagawa\thanks{Corresponding author at: Graduate School of Economics, The University of Tokyo, 7-3-1 Hongo, Bunkyo-ku, Tokyo 113-0033, Japan. Email:$\ $hayato-tagawa@g.ecc.u-tokyo.ac.jp.}}
\date{\today}

\maketitle
\vspace{-0.6cm}
\begin{abstract}
The synthetic control method (SCM) is widely used for causal inference with panel data, particularly when the number of treated units is small. It relies on the stable unit treatment value assumption (SUTVA), ruling out spillover effects. However, interventions often affect not only treated but also untreated units. This study proposes a novel panel data method that extends standard SCM to account for spillovers and estimate both treatment and spillover effects. The approach extends the SCM framework by incorporating a spatial autoregressive (SAR) panel data model that captures spillover patterns across units. We also develop a Bayesian inference procedure using horseshoe priors for regularization. We apply the proposed method to two empirical studies: (i) evaluating the effect of the California tobacco tax on cigarette consumption, and (ii) assessing the economic impact of the 2011 Sudan division on GDP per capita.
\bigskip \\
\textbf{Keywords:} Bayesian inference; national breakup; spatial autoregressive model; spatial panel data; SUTVA
\end{abstract}


\newpage

\newpage

\setstretch{1.4} 

\section{Introduction}

The synthetic control method (SCM) \citep{abadie2003economic,abadie2010synthetic} is a causal inference approach used to estimate treatment effects when the number of units undergoing intervention is very small, such as a single unit. 
This approach identifies treatment effects from panel data by substituting the counterfactual outcomes of the treated unit in the absence of treatment with weighted averages of outcomes of untreated units.
The method has been widely used in various fields, including political economy and marketing.
\citet{athey2017state} describe SCM as ``arguably the most important innovation in the policy evaluation literature in the last 15 years.'' 

The SCM is typically applied to country- or district-level data.
For instance, \cite{Abadie_et_al_2015} apply SCM to country-level data to estimate the economic impact of the 1990 German reunification, while \cite{bifulco2017using} apply it to school district data to evaluate the effect of an educational program.
In such contexts, interventions may affect untreated units through geographic or socioeconomic connections among countries or districts, which are known as \textit{spillovers}. However, SCM relies on the stable unit treatment value assumption (SUTVA) \citep{rubin1978bayesian}, which posits that one unit's outcome is unaffected by the treatment status of other units, excluding spillovers. If spillovers occur to untreated units, SUTVA is violated, which can bias treatment effect estimation by the SCM.

This paper tackles this problem by extending the standard SCM to allow for spillover effects. 
We leverage the spatial autoregressive (SAR) model, a workhorse model in spatial data analysis, to address this issue with its capability to capture dependence of outcomes among units. We propose a novel approach to identify and estimate treatment effects by incorporating an SAR panel data model into SCM, where we characterize the outcomes of untreated units with the SAR model. This approach allows for spillover effects arising from the dependence of outcomes among treated and untreated units, thereby relaxing SUTVA in SCM. Furthermore, the extended SCM can identify and estimate spillover effects on untreated units, which are also parameters of interest in many empirical studies adopting SCM.

Building on the identification results, we propose a Bayesian inference method for the proposed SCM. Data targeted by SCM often involves panel data, but the number of units or length of pretreatment periods is not always sufficiently large. In such cases, frequentist approaches may not provide accurate inference. Bayesian inference offers several advantages. First, it can yield more accurate estimates than frequentist methods when the sample size is small. Second, it facilitates statistical inference via the Markov Chain Monte Carlo (MCMC) procedure. Third, Bayesian modeling is flexible, allowing models to be less dependent on the assumptions about prior distributions.

We apply Bayesian regularization in the construction of synthetic controls. Following \cite{kim2020bayesian}, we use Bayesian horseshoe priors \citep{carvalho2010horseshoe} to model the synthetic weights. The Bayesian horseshoe prior has strong regularization effects, making it effective in avoiding overfitting bias. This is particularly important in SCM, as it is often applied to data with a relatively large number of control units and short pretreatment periods, which can easily cause overfitting bias in the estimation of synthetic control outcomes. The simulation study in the paper examines the finite sample performance of the proposed Bayesian inference.

This paper presents two substantive empirical studies applying the proposed SCM. The first revisits \citet{abadie2010synthetic} to evaluate the impact of California’s tobacco tax on cigarette consumption, and the second estimates the effect of Sudan’s north–south split in 2011 on GDP per capita. Both studies involve spillover effects among US states and African countries, respectively. Regarding California's tobacco tax, our results show a negative impact on tobacco consumption in California, supporting the findings of \cite{abadie2010synthetic}. We also find evidence of spillover effects, showing that the tax reduced consumption in other states.

To the best of our knowledge, the second empirical study is the first to examine the economic impact of Sudan's north-south split on the Sudans (the region of the former united Sudan). In a related study, \citet{mawejje2021economic} estimate the economic impact of the 2012 oil production halt in South Sudan on South Sudan itself using the standard SCM. They intentionally exclude countries neighboring South Sudan from the control group to address spillover concerns. This is a common practice in dealing with spillover in SCM \citep{abadie2021using};
however, excluding untreated units, particularly neighboring countries, can lead to poor fitting of synthetic controls and cause additional bias. The SCM proposed in this paper does not need to exclude any untreated units. Our empirical results show that Sudan's north-south split led to a 9.5\% decline in GDP per capita in the Sudans, with a cumulative reduction of 34\% between 2011 and 2015. We also find evidence of negative spillover effects on other African countries with strong economic ties to Sudan, such as Egypt and Kenya.


\subsection*{Related Literature}
Since \citet{abadie2003economic} and \citet{abadie2010synthetic} pioneered SCM, along with its broad application in the social sciences, SCM has been theoretically and methodologically advanced by many works. 
\citet{abadie2021penalized} present a penalized synthetic control method to reduce interpolation biases. 
\citet{arkhangelsky2021synthetic} combine SCM with difference-in-differences (DID) and propose a new estimator with robustness properties.
\citet{ben2021augmented} propose an augmented synthetic control method to correct bias resulting from imperfect pre-treatment fit. 
\citet{li2020statistical} and \citet{chernozhukov2021exact} propose inference methods for SCM.
\citet{ferman_pinto_2021} point out a potential bias in SCM when the perfect fit assumption is not satisfied.

Several works have proposed Bayesian estimation in SCM to facilitate statistical inference through MCMC sampling and/or to apply flexible modeling. \cite{kim2020bayesian} introduce a Bayesian synthetic control method that uses the Bayesian horseshoe prior proposed by \cite{carvalho2010horseshoe} for the synthetic weights. Their simulations demonstrate more accurate predictions of counterfactual outcomes than the standard SCM, particularly when the number of units is relatively large compared to the length of the pretreatment periods. 
\cite{brodersen2015inferring} propose a method based on Bayesian structural time-series models. 
\cite{pang2022bayesian} and \cite{klinenberg2023synthetic} propose approaches to correct bias in SCM induced by the dynamic characteristics of untreated units. \cite{pang2022bayesian} propose a Bayesian posterior predictive approach with a latent factor term that incorporates unit-specific time trends. \cite{klinenberg2023synthetic} incorporates time-varying parameters by using a state space framework.

While the related literature mentioned so far assumes SUTVA in SCM, few studies have considered its relaxation. \cite{cao2019estimation} allow for spillover effects in SCM by assuming a specified structure of spillover effects. Our approach differs from theirs in two aspects: (i) our approach supposes a spatial correlation structure for the outcomes of units rather than a direct structure for spillover effects, and (ii) our approach supposes the existence of a synthetic control only for the treated unit, whereas their approach supposes synthetic controls for both the treated and untreated units. \cite{menchetti2022estimating} assume that a unit's potential outcome depends on its own and neighboring treatment assignments. They propose a Bayesian structural time series method to identify and estimate treatment and spillover effects by estimating the predictive distribution of counterfactual outcomes. \cite{grossi2020synthetic} focus on the average spillover effects within specified clusters and provide a method for constructing confidence intervals for treatment and average spillover effects. In the context of DID, \cite{miguel2004worms} propose an approach for estimating spillover effects using RCT data, which utilizes non-compliance in treatment groups to identify spillover effects.

This work also relates and contributes to the economic analysis of national breakups \citep[e.g.,][]{alesina2000economic,bolton1997breakup} through its original empirical study.
The economic consequences of major political reconfigurations, such as national breakups or unifications, have long been the subject of rigorous empirical and theoretical research. For example, \citet{redding2008costs} exploit the division and reunification of Germany as a natural experiment to quantify the impact of border changes on the economic development of border regions. Seminal works by \citet{alesina2000economic} and \citet{bolton1997breakup} explore the interplay between economic integration, political economy conflicts, and the formation or dissolution of nations, providing foundational frameworks for understanding such large-scale transformations. 
This paper contributes to this literature by empirically demonstrating the macroeconomic impacts of national division, using unique data on Sudan's split and a novel SCM approach that accounts for spillover effects arising from the breakup.

\subsection*{Structure of the Paper}
The remainder of the paper proceeds as follows. Section \ref{sec:setup} describes the SCM setup with spillover effects.
Section \ref{sec:identification} introduces the SAR panel data model and presents the identification results.
Section \ref{sec:estimation} proposes a Bayesian SCM procedure. Section \ref{sec:simulation_study} shows simulation results to demonstrate the finite-sample performance of the proposed SCM.
Section \ref{sec:application} presents the results of two empirical applications.
Section \ref{sec:conclusion} concludes the paper. The supplementary appendix provides additional details of the Bayesian SCM procedure.


\section{Setup}\label{sec:setup}

Consider units $i=0,1,\cdots,N$ and time points $t=1,2,\cdots,T$.
We suppose that $i=0$ is the unit receiving the treatment, and $i=1,\cdots,N$ are the units in the control group.
Let $T_{0}$ be the number of pretreatment periods with $1 \leq T_0 < T$.
The treatment status of a unit $i$ at time $t$ is denoted as $d_{it} \in \{0,1\}$, where $d_{it}=1$ means that unit $i$ receives treatment at time $t$, and $d_{it}=0$ means otherwise. In our context, $d_{it}=1$ when $i=0$ and $t > T_0$ and $d_{it}=0$ otherwise. 
We define $\bm{d}_{t} \equiv (d_{0t}, d_{1t},\cdots,d_{Nt})\in\{0,1\}^{N+1}$, the treatment status vector for time $t$.

We consider the potential outcome framework. While many studies assume that one's potential outcome is not influenced by the treatment status of others (SUTVA) \citep{rubin1978bayesian}, this study allows for the potential outcome to be affected by the treatment status of others.
Thus, for treatment status $\bm{d}_{t}$ at time $t$, the potential outcome of unit $i$ is denoted as $Y_{it}(\bm{d}_{t})$.\footnote{SUTVA means that the potential outcome of each unit $i$ is represented by $Y_{it}(d_{it})$, which does not depend on the treatment status of others $\bm{d}_{t} \backslash d_{it}$,}
The observed outcomes are
\begin{align*}
  Y_{it} \equiv \begin{cases}
             Y_{it}(0,0,\cdots,0) & \text{if } t\leq T_{0}, \\
             Y_{it}(1,0,\cdots,0) & \text{if } t>T_{0}
           \end{cases}
\end{align*}
for all $i=0,1,\cdots,N$ and $t=1,2,\ldots,T$.

We define the \textit{treatment effect} $\xi_{0t}$ for the treated unit $0$ at time $t$ $(> T_{0})$ as follows:
\begin{align}
  \xi_{0t} \equiv \underbrace{Y_{0t}(1,0,\cdots,0)}_{\text{observed outcome}} - \underbrace{Y_{0t}(0,0,\cdots,0)}_{\text{counterfactual outcome}},\nonumber
\end{align}
where $Y_{0t}(1,0,\cdots,0)$ is the outcome when unit $0$ is treated, and $Y_{0t}(0,0,\cdots,0)$ is the outcome when the unit is not treated, a counterfactual outcome not observed in the data.

The framework of this study can capture the spillover effects on units in the control group.
Although no unit in the control group receives treatment, the influence of unit $0$ receiving treatment may affect the outcomes of the control group units, which is referred to as spillover effects.
We define the \textit{spillover effect} $\xi_{it}$ for any untreated unit $i\ (\geq 1)$ and time $t\ (>T_{0})$ as follows:
\begin{align}
  \xi_{it} \equiv \underbrace{Y_{it}(1,0,\cdots,0)}_{\text{observed outcome}} - \underbrace{Y_{it}(0,0,\cdots,0)}_{\text{counterfactual outcome}}.\nonumber
\end{align}
The spillover effect is the counterfactual difference in the outcomes of an untreated unit when unit $i=0$ receives treatment versus when unit $i=0$ does not. 

For the outcomes of the treated unit, following \citet{abadie2010synthetic}, we suppose the following assumption.

\medskip
\begin{assumption}[Perfect Fit]\label{Perfect fit}
There exists a vector of weights $\bm{\alpha} = (\alpha_{1},\alpha_{2},\cdots,\alpha_{N})^{\top}\in\mathbb{R}^{N}$ that satisfies the following: For each $t=1,2,\cdots,T$,
  \begin{align}
    Y_{0t}(0,0,\cdots,0) = \sum_{i=1}^{N}\alpha_{i}Y_{it}(0,0,\cdots,0) \mbox{\ \ a.s.} \label{eq:perfect_fit}
  \end{align}
\end{assumption}
\medskip
This assumption implies that, in the absence of treatment, the control outcome of the treated unit can be replicated by the weighted average of the control outcomes of the untreated units.  This is a fundamental assumption for SCM and is either explicitly or implicitly imposed in most SCM studies.
Since our framework allows outcomes to depend on the treatment status of other units, Assumption \ref{Perfect fit} pertains to a state where all units are untreated.

Under Assumption \ref{Perfect fit}, we can estimate the synthetic weights $\bm{\alpha}$ by solving the following least-squares problem:
\begin{align}
  \widehat{\bm{\alpha}} = \argmin_{\bm{\alpha}\in\mathbb{R}^{N}} \sum_{t=1}^{T_{0}} \bigg(Y_{0t}-\sum_{i=1}^{N}\alpha_{i}Y_{it}\bigg)^{2}. \label{estimate alpha}
\end{align}
Using $\widehat{\bm{\alpha}}$, the standard SCM \citep{abadie2010synthetic} estimates the treatment effect $\xi_{0t}$ ($t > T_0$) by $Y_{0t} - \sum_{i=1}^{N}\hat{\alpha}_{i}Y_{it}$. 
Under SUTVA and Assumption \ref{Perfect fit} (perfect fit), \citet{abadie2010synthetic} show the identification and unbiased estimation of treatment effects using SCM.
However, when SUTVA is violated, the standard SCM can be biased in its estimation of treatment effects. The subsequent section clarifies the sources of this bias and introduces a new approach to identify both treatment and spillover effects.


\section{Identification}\label{sec:identification}

This section begins by clarifying the bias in the standard SCM that arises from spillover effects. We then introduce the SAR panel data model and discuss its capability to account for spillover effects. Finally, we present our main identification result for both treatment and spillover effects.


\subsection{Bias for the Standard SCM}

We first show that when SUTVA is violated, the standard SCM can be biased due to spillover.
Assumption \ref{Perfect fit} implies that the treatment effects $\xi_{0t}$ can be expressed as $Y_{0t}(1,0,\cdots,0)-\sum_{i=1}^{N}\alpha_{i}Y_{it}(0,0,\cdots,0)$.
However, when the outcomes depend on the treatment statuses of other units, we cannot observe counterfactual control outcomes $Y_{it}(0,0,\cdots,0)$ after the pre-treatment periods ($t > T_0$). This causes bias in the standard SCM as follows:
\begin{align*}
  Y_{0t} - \sum_{i=1}^{N}\alpha_{i}Y_{it} &=  Y_{0t}(1,0,\cdots,0) - \sum_{i=1}^{N}\alpha_{i}Y_{it}(1,0,\cdots,0) \\
  & = \xi_{0t}+\underbrace{\sum_{i=1}^{N}\alpha_{i}\big(Y_{it}(0,0,\cdots,0) - Y_{it}(1,0,\cdots,0)\big)}_{\text{bias}},
\end{align*}
where $Y_{0t} - \sum_{i=1}^{N}\alpha_{i}Y_{it}$ is the standard SCM estimator of $\xi_{0t}$ given knowledge of  $\bm{\alpha}$. This result shows the existence of bias in the standard SCM when SUTVA is not satisfied.


\subsection{Spatial Autoregressive Model}

To address the issue of bias caused by spillover, we introduce a model that captures spillover effects between treated and untreated units. Specifically, for the outcomes of units in the control group, we assume the following SAR panel data model for each $t\geq 1$:
\begin{align}
  \underbrace{\bm{Y}_{t}^{c}(\bm{d}_{t})}_{N \times 1} = \rho \big(\bm{w} Y_{0t}(\bm{d}_{t}) + \bm{W} \bm{Y}_{t}^{c}(\bm{d}_{t})\big) + \bm{X}_{t}\bm{\beta} + \bm{u}_{t}^{\bm{d}_{t}}, \label{network model}
\end{align}
where $\bm{Y}_{t}^{c}(\bm{d}_{t}) \equiv (Y_{1t}(\bm{d}_{t}), Y_{2t}(\bm{d}_{t}),\cdots, Y_{Nt}(\bm{d}_{t}))^{\top}\in\mathbb{R}^{N}$.
The vector $\bm{w}\in\mathbb{R}^{N}$ and matrix $\bm{W}\in\mathbb{R}^{N\times N}$ comprise spatial weights,  
which are to be specified a priori. For $i=1,2,\ldots,N$ and $j=0,1,\ldots,N$, we denote by $w_{ij}$ the spatial weight between units $i$ and $j$; that is, the $(i,j+1)$-th element of the matrix $(\bm{w},\bm{W}) \in \mathbb{R}^{N\times (1+N)}$. Typical examples of spatial weights include adjacent weights, where $w_{ij}$ is 1 if units $i$ and $j$ are adjacent and 0 otherwise. Geographic distance (e.g., distance between the capitals of two countries) and economic distance (e.g., trade amount between two countries) are also often employed as spatial weights in the literature on spatial data analysis \citep{lesage2009introduction}. 
The matrix $\bm{X}_{t}\equiv (\bm{X}_{1t},\ldots,\bm{X}_{Nt})^{\top}\in\mathbb{R}^{N\times k}$ denotes the covariate matrix for control units, where $\bm{X}_{it}$ is the covariate vector for unit $i$ at time $t$. We suppose that $\bm{X}_{t}$ does not depend on the treatment status $\bm{d}_{t}$.  The vector $\bm{u}_{t}^{\bm{d}_{t}} \equiv (u_{1t}^{\bm{d}_{t}},\cdots,u_{Nt}^{\bm{d}_{t}})^{\top}\in\mathbb{R}^{N}$ contains the error terms.

The SAR panel data model \eqref{network model} captures spillover effects arising from spatial dependence in outcomes between treated and untreated units. 
The magnitude of these spillovers is determined by the spatial autoregressive coefficient $\rho$ as well as the spatial weights $\bm{w}$ and $\bm{W}$, which encode the strength and structure of spatial correlations.
Several methods have been proposed to estimate $\rho$ in  \modelname panel data models  \citep[e.g.,][]{fingleton2008generalized,lee2010estimation,su2012semiparametric,glass2016spatial,liang2022semiparametric}.


\subsection{Identification}
We now turn to the identification of treatment and spillover effects under Assumption \ref{Perfect fit} and the SAR panel data model \eqref{network model}.
For notational simplicity, since the treatment state $d_{it}$ for all $i \geq 1$ and $t$ is always $0$, we write
\begin{align}
  Y_{it}(1) & = Y_{it}(1,0,\cdots,0) \ \ \ \ (\forall i,t), \nonumber \\
  Y_{it}(0) & = Y_{it}(0,0,\cdots,0) \ \ \ \ (\forall i,t). \nonumber
\end{align}
We impose the following assumption.
\medskip
\begin{assumption}\label{individual factor}
For $t=1,2,\cdots,T$ and $i=1,2,\cdots,N$, the error term $u_{it}^{\bm{d}_{t}}$ depends only on its own treatment state $d_{it}$; that is, $u_{it}^{\bm{d}_{t}} = u_{it}^{\bm{d}_{t}^{\prime}}$ a.s. for any $\bm{d}_{t}=(d_{0t},d_{1t},\ldots,d_{1T})^{\top}$ and $\bm{d}_{t}^{\prime}=(d_{0t}^{\prime},d_{1t}^{\prime},\ldots,d_{1T}^{\prime})^{\top}$ such that $d_{it} = d_{it}^{\prime}$.
\end{assumption}
\medskip

This assumption rules out spillover effects arising from unobserved factors, focusing instead on those generated by outcome dependence. Addressing the former would require additional structural assumptions that may not align with the standard SCM framework.
Given this assumption and the fact that $d_{it}=0$ for any control unit $i$ at any time $t$, the error terms for the control units always satisfy $\bm{u}_{t}^{\bm{d}_{t}}=\bm{u}_{t}^{\bm{0}_{N+1}}$, where $\bm{0}_{N+1}$ denotes an $(N+1)$-dimensional zero vector.  
For notational simplicity, we denote $\bm{u}_{t}^{\bm{d}_{t}}$($=\bm{u}_{t}^{\bm{0}_{N+1}}$) by $\bm{u}_{t}$.

We further assume that the model (\ref{network model}) and the synthetic weights $\bm{\alpha}$ in Assumption \ref{Perfect fit} satisfy the following condition.

\medskip
\begin{assumption}\label{invertible} $\bm{I}_{N}-\rho \bm{w}\bm{\alpha}^{\top} -\rho \bm{W}$ is full rank.
\end{assumption}
\medskip

This assumption ensures that the matrix $\bm{I}_{N}-\rho \bm{w}\bm{\alpha}^{\top} -\rho \bm{W}$ is invertible. Assumption \ref{invertible} is testable given estimators of $\bm{\alpha}$ and $\rho$. 

Given the parameters $\bm{\alpha}$ and $\rho$, the treatment and spillover effects can be identified as follows.
For each $t>T_{0}$, under Assumption \ref{Perfect fit} and the SAR panel data model \eqref{network model}, we have
\begin{align}
  \bm{Y}_{t}^{c}(0) & = \rho \bm{w} Y_{0t}(0) + \rho \bm{W} \bm{Y}_{t}^{c}(0) + \bm{X}_{t}\bm{\beta} + \bm{u}_{t} \nonumber                           \\
                    & =\rho \bm{w} \bm{\alpha}^{\top} \bm{Y}_{t}^{c}(0) + \rho \bm{W} \bm{Y}_{t}^{c}(0) + \bm{X}_{t}\bm{\beta} + \bm{u}_{t}.\nonumber
\end{align}
Since the matrix $(\bm{I}_{N}-\rho \bm{w}\bm{\alpha}^{\top} -\rho \bm{W})$ is invertible under Assumption \ref{invertible}, rearranging the above equation yields
\begin{align}
  \bm{Y}_{t}^{c}(0) & = \big(\bm{I}_{N}-\rho \bm{w}\bm{\alpha}^{\top} -\rho \bm{W}\big)^{-1}(\bm{X}_{t}\bm{\beta} + \bm{u}_{t}).\nonumber
\end{align}
The treatment effect can then be written as
\begin{align}
  \xi_{0t} & = Y_{0t}(1) - \bm{\alpha}^{\top} \bm{Y}_{t}^{c}(0) \nonumber                                                                                              \\
           & =Y_{0t}(1) - \bm{\alpha}^{\top} \big(\bm{I}_{N}-\rho \bm{w}\bm{\alpha}^{\top} -\rho \bm{W}\big)^{-1}\big(\bm{X}_{t}\bm{\beta} + \bm{u}_{t}\big).\nonumber
\end{align}
Furthermore, under the model \eqref{network model} and Assumption, \ref{individual factor}, we have for each $t > T_0$:
\begin{align}
  \big(\bm{I}_{N} - \rho \bm{W}\big)\bm{Y}_{t}^{c}(1) - \rho \bm{w}Y_{0t}(1)= \bm{X}_{t}\bm{\beta} + \bm{u}_{t}.\nonumber
\end{align}
Therefore, for $t > T_0$, we obtain
\begin{align}
  \xi_{0t} & =Y_{0t}(1) - \bm{\alpha}^{\top} \big(\bm{I}_{N}-\rho \bm{w}\bm{\alpha}^{\top} -\rho \bm{W}\big)^{-1}\big(\big(\bm{I}_{N} - \rho \bm{W}\big)\bm{Y}_{t}^{c}(1) - \rho \bm{w}Y_{0t}(1)\big) \nonumber \\
  &= Y_{0t} - \bm{\alpha}^{\top} \big(\bm{I}_{N}-\rho \bm{w}\bm{\alpha}^{\top} -\rho \bm{W}\big)^{-1}\big(\big(\bm{I}_{N} - \rho \bm{W}\big)\bm{Y}_{t}^{c} - \rho \bm{w}Y_{0t}\big). \label{identification}
\end{align}
If the parameters $\bm{\alpha}$ and $\rho$ are estimated using data from the pretreatment periods ($t\leq T_{0}$), then $\xi_{0t}$ ($t> T_{0}$) can be estimated as equation (\ref{identification}). Note that $\bm{\alpha}$ can be consistently estimated by the least-squares method (\ref{estimate alpha}) under Assumption \ref{Perfect fit},  and several methods are available to estimate 
$\rho$ in the literature of the SAR panel data model \citep[e.g.,][]{fingleton2008generalized,lee2010estimation,su2012semiparametric,glass2016spatial,liang2022semiparametric}. 

The following theorem summarizes the identification result for $\xi_{0t}$.
\medskip

\begin{theorem}[Identification of the Treatment Effect]\label{Identification for Treatment Effect}
Suppose that Assumptions \ref{Perfect fit}, \ref{individual factor}, and \ref{invertible} hold. Then, given $\rho$ and $\bm{\alpha}$, the treatment effect $\xi_{0t}$ for $t>T_{0}$ can be identified as follows:
  \begin{align}
    \xi_{0t} = Y_{0t} - \bm{\alpha}^{\top} \big(\bm{I}_{N}-\rho \bm{w}\bm{\alpha}^{\top} -\rho \bm{W}\big)^{-1}\big(\big(\bm{I}_{N} - \rho \bm{W}\big)\bm{Y}_{t}^{c} - \rho \bm{w}Y_{0t}\big). \label{eq:identification_treatment_effect}
  \end{align}
\end{theorem}

\begin{proof}
    The result follows from the discussion above.
\end{proof}

\medskip
The discussion above shows that the counterfactual outcome $\bm{Y}_{t}^{c}(0)$ for each untreated unit for $t>T_{0}$ can be identified as $\big(\bm{I}_{N}-\rho \bm{w}\bm{\alpha}^{\top} -\rho \bm{W}\big)^{-1}\big(\big(\bm{I}_{N} - \rho \bm{W}\big)\bm{Y}_{t}^{c} - \rho \bm{w}Y_{0t}\big)$. Therefore, given the parameters $\bm{\alpha}$ and $\rho$, the spillover effects on the untreated units are also identifiable.

\medskip
\begin{theorem}[Identification of the Spillover Effects]\label{Identification for Spillover Effect}
Suppose that Assumptions \ref{Perfect fit}, \ref{individual factor}, and \ref{invertible} hold. Then, given $\rho$ and $\bm{\alpha}$, the spillover effects $\bm{\xi}_{t}^{c}=(\xi_{1t},\xi_{2t},\cdots,\xi_{Nt})^{\top}\in\mathbb{R}^{N}$ on the $N$ control units for $t>T_{0}$ can be identified as follows:
  \begin{align}
    \bm{\xi}_{t}^{c} & =\bm{Y}_{t}^{c} - \big(\bm{I}_{N}-\rho \bm{w}\bm{\alpha}^{\top} -\rho \bm{W}\big)^{-1}\big(\big(\bm{I}_{N} - \rho \bm{W}\big)\bm{Y}_{t}^{c} - \rho \bm{w}Y_{0t}\big).\label{eq:identification_spillover_effect}
  \end{align}
\end{theorem}

\begin{proof}
    Note that $\bm{Y}_{t}^{c} = \bm{Y}_{t}^{c}(1)$ for $t>T_0$ and that $\big(\bm{I}_{N}-\rho \bm{w}\bm{\alpha}^{\top} -\rho \bm{W}\big)^{-1}\big(\big(\bm{I}_{N} - \rho \bm{W}\big)\bm{Y}_{t}^{c} - \rho \bm{w}Y_{0t}\big) = \bm{Y}_{t}^{c}(0)$. This leads to the result (\ref{eq:identification_spillover_effect}).
\end{proof}

\medskip

Theorems \ref{Identification for Treatment Effect} and \ref{Identification for Spillover Effect} show that by incorporating the SAR model into the SCM framework, treatment and spillover effects can be identified without relying on SUTVA. 
The results in Theorems \ref{Identification for Treatment Effect} and \ref{Identification for Spillover Effect} also suggest that although the SAR model (\ref{network model}) includes the covariates $\bm{X}_t$ and error terms $\bm{u}_{t}$, it is not necessary to estimate $\bm{\beta}$ or the distribution of $\bm{u}_{t}$ to estimate treatment and spillover effects. For their estimation, only the spatial autoregressive parameter $\rho$ and the spatial weights $(\bm{w},\bm{W})$ in the SAR model (\ref{network model}) are relevant.

Remarkably, when there is no spillover (i.e., $\rho=0$), our identification result (\ref{eq:identification_treatment_effect}) for $\xi_{0t}$ corresponds to the identification result of the standard SCM \citep{abadie2010synthetic}. That is, when $\rho=0$,  the identification result (\ref{eq:identification_treatment_effect}) simplifies to
\begin{align*}
  \xi_{0t} & = Y_{0t} - \bm{\alpha}^{\top} (\bm{I}_{N}-\rho \bm{w}\bm{\alpha}^{\top} - \rho \bm{W})^{-1}\big((\bm{I}_{N}-\rho \bm{W})\bm{Y}_{t}^{c} - \rho \bm{w} Y_{0t}\big) \\
           & = Y_{0t} - \bm{\alpha}^{\top} \bm{I}_{N}(\bm{Y}_{t}^{c}-0)                                                                                 \\
           & = Y_{0t} - \bm{\alpha}^{\top} \bm{Y}_{t}^{c},
\end{align*}
where the final line is identical to the identification equation for the standard SCM \citep{abadie2010synthetic}.
This result indicates that the identification result obtained from the standard SCM is the special case of our identification result when $\rho=0$ (no spillover).


\section{Inference}\label{sec:estimation}

In this section, we propose a Bayesian inference method for treatment effects $\xi_{0t}$ and spillover effects $\xi_{it}$ ($i =1,2,\ldots,N$), building on the identification results presented in Section \ref{sec:identification}. SCM is typically applied to data with small or no large pretreatment periods, where frequentist methods may not provide accurate estimations. On the other hand, Bayesian methods can provide precise statistical inference (e.g., credible intervals) through MCMC procedures even with short pretreatment periods. Moreover, by examining the posterior samples of $\rho$, one can test for the presence of spatial correlation.

This section presents a Bayesian inference method for the model described in Sections \ref{sec:setup} and \ref{sec:identification}.
Let $\bm{\theta}= (\bm{\alpha},\rho,\bm{\beta},\bm{\eta},\{\bm{\gamma}_{t}\}_{t=1}^{T_{0}},\sigma^{2})$ denote the parameter vector.  
Under the model in Sections~\ref{sec:setup}--\ref{sec:identification}, the pre-treatment outcomes satisfy
\begin{align*}
Y_{0t} &= \bm{\alpha}^\top \bm{Y}_t^{c}, \\
(\bm{I}_{N} - \rho \bm{W}) \bm{Y}_t^{c}
    &= \rho\, \bm{w} Y_{0t} + \bm{X}_t \bm{\beta} + \bm{\eta}\bm{\gamma}_t + \bm{e}_t,
\qquad t = 1,\ldots,T_0.
\end{align*}
Letting $\tilde{\bm{u}}_t$ denote the transformed disturbance term, the conditional density of $\bm{Y}_t^c$ given 
$(Y_{0t},\rho,\bm{\beta},\bm{\eta},\bm{\gamma}_t,\sigma^2)$ is Gaussian with Jacobian factor 
$\lvert \bm{I}_N - \rho \bm{W} \rvert$.  
Thus, the joint likelihood for the pre-treatment sample is
\begin{align}
\label{eq:fulljointlik}
\mathcal{L}(\bm{\theta} \mid \{Y_{0t},\bm{Y}_t^{c}\}_{t\le T_0})
    &= \prod_{t=1}^{T_0}
        p\!\left(Y_{0t} \mid \bm{\alpha},\bm{Y}_t^{c}\right)
        \cdot
        \lvert \bm{I}_N - \rho \bm{W} \rvert
        (2\pi\sigma^{2})^{-N/2}
        \exp\!\left\{-\frac{1}{2\sigma^{2}}
        \tilde{\bm{u}}_t^\top \tilde{\bm{u}}_t\right\}.
\end{align}

Expression (\ref{eq:fulljointlik}) makes explicit that (i) information from the treated unit $Y_{0t}$ enters entirely through the synthetic-control restriction, while (ii) the spatial autoregressive structure governs the distribution of $\bm{Y}_t^{c}$.  
However, the likelihood contains the multiplicative interaction term $\rho\,\bm{w}\bm{\alpha}^\top$, which produces weak identification of $(\bm{\alpha},\rho)$ and leads to extremely slow mixing when implementing full joint MCMC.  
Furthermore, jointly sampling all latent factor components $(\bm{\eta},\bm{\gamma}_t)$ over $T_0$ periods imposes substantial computational burden without improving post-treatment inference.  
In practice, we found that full joint estimation results in unstable chains, poor effective sample sizes, and unreliable inference for the treatment effects.
To overcome these identification and computational difficulties, we decompose the joint likelihood in (\ref{eq:fulljointlik}) into two components,
\begin{align}
\mathcal{L}_{1}(\bm{\alpha} \mid Y_{0},Y^{c})
    &= \prod_{t=1}^{T_0}
        p\!\left(Y_{0t} \mid \bm{\alpha}, \bm{Y}_t^{c}\right),
\label{eq:L1}
\\[0.3em]
\mathcal{L}_{2}(\rho \mid Y^{c}, \hat{\bm{\alpha}})
    &= \prod_{t=1}^{T_0}
        p\!\left(\bm{Y}_t^{c} \mid \rho, \hat{\bm{\alpha}}\right),
\label{eq:L2}
\end{align}
and estimate the two components sequentially.
We therefore adopt a two-step Bayesian inference strategy in which (i) the synthetic control weights $\bm{\alpha}$ are sampled from \eqref{eq:L1}, and the spatial parameter $\rho$ is sampled from \eqref{eq:L2} conditional on $\hat{\bm{\alpha}}$.
The detailed MCMC procedures for each step are described in the following subsections.
The validity of the MCMC implementation including the joint distribution test of \citet{geweke2004getting}, prior predictive checks, sensitivity analyses, and convergence diagnostics is documented in detail in the Supplementary Appendix.

\subsection{Synthetic Weights}

Regarding the estimation of the synthetic weights $\bm{\alpha}$, we adopt \citeauthor{kim2020bayesian}'s (\citeyear{kim2020bayesian}) method, which utilizes the Bayesian horseshoe prior as the prior distribution for parameters.
The Bayesian horseshoe prior places a hierarchical prior distribution on parameters, characterized by a strong shrinkage effect around zero.
\citet{park2008bayesian} show that LASSO regression in linear models corresponds to estimation with a Laplace prior distribution on coefficients, while the Bayesian horseshoe prior proposed by \citet{carvalho2010horseshoe} suggests an even more selective coefficient choice.
Thus, by placing this prior distribution on the synthetic weights, it is anticipated that the estimated weights will select the units in the control group that are more related to the treated unit.

The following hierarchical prior distributions are placed on $\alpha_{1},\alpha_{2},\cdots,\alpha_{N}$:
\begin{align*}
  \alpha_{i} \mid \lambda_{i} & \sim N(0,\lambda_{i}^{2}),\ \ \text{for}\ i=1,2,\cdots,N,       \\
  \lambda_{i} \mid \tau       & \sim \text{Half-Cauchy}(0,\tau),\ \ \text{for}\ i=1,2,\cdots,N, \\
  \tau \mid \sigma_{1}        & \sim \text{Half-Cauchy}(0,\sigma_{1}),                          \\
  \sigma_{1}                  & \sim \text{Half-Cauchy}(0,10).
\end{align*}
When these prior distributions are set, as shown by \citet{makalic2015simple}, the full conditionals of each parameter can be derived analytically, enabling parameter sampling using the Gibbs sampler. Full conditional distributions for each parameter are presented in the supplementary appendix.


\subsection{Spatial Autoregressive Panel Data Model}
Regarding the model (\ref{network model}), we set the Bayesian horseshoe prior for $\bm{\beta}$ as in estimating the synthetic weights, that is,\begin{align*}
  \beta_{k}\mid\kappa_{k} & \sim N(0,\kappa_{i}^{2}), \ \ \text{for}\ k=1,2,\cdots,K,       \\
  \kappa_{k}\mid\psi      & \sim \text{Half-Cauchy}(0,\psi),\ \ \text{for}\ k=1,2,\cdots,K, \\
  \psi\mid\sigma_{2}      & \sim \text{Half-Cauchy}(0,\sigma_{2}),                          \\
  \sigma_{2}              & \sim \text{Half-Cauchy}(0,10).
\end{align*}


We assume that $\bm{u}_{t}$ follows a multilevel latent factor model $\bm{u}_{t}=\bm{\eta}\bm{\gamma}_{t}+\bm{e}_{t}$, where $\bm{\gamma}_{t}\in\mathbb{R}^{p}$ denote factors at time $t$,  $\bm{\eta}\in\mathbb{R}^{N \times p}$ are factor loadings, and $\bm{e}_{t} \in\mathbb{R}^{N}$ is an error vector.
We assume $\bm{\gamma}_{t}$ follow the AR($1$) model $\bm{\gamma}_{t}=\phi_{\gamma}\bm{\gamma}_{t-1} + \bm{v}_{t}$, where $\bm{v}_{t}\sim N(\bm{0}_{p},\sigma^{2}_{\gamma}\bm{I}_{p})$.
For each control unit $i$ ($=1,\cdots,N$), $\eta_{i}$ follows $p$-dimensional normal distribution $N_{p}(\bm{0}_{p},\sigma^{2}_{\eta}\bm{\Sigma}_{\eta})$, where $\bm{\Sigma}_{\eta}=\mathrm{diag}(\omega_{1}^{2},\cdots,\omega_{p}^{2})$. The priors of  $\sigma_{\eta}$ and $\omega_{1}, \cdots,\omega_{p}$ are $\text{half-Cauchy}(0,10)$. These settings are the same as those used by \citet{pang2022bayesian}.

We adopt the method proposed by \cite{lesage1997bayesian} and \citet{lesage2009introduction} for modeling $\rho$. 
Given all parameters except $\rho$, the conditional distribution of $\rho$ is
\begin{align*}
  p(\rho\mid\text{rest}) \propto \prod_{t=1}^{T_{0}}|\bm{I}_{N} - \rho\bm{W}|\exp(-\dfrac{1}{2\sigma^{2}}\tilde{\bm{u}}_{t}^{\top}\tilde{\bm{u}}_{t}),
\end{align*}
where $\tilde{\bm{u}}_{t} = \bm{A}\bm{Y}_{t}^{c} - \rho\bm{w}Y_{0t}-\bm{X}_{t}\bm{\beta} - \bm{\gamma}_{t}^{\top}\bm{\eta}$ and $\bm{A}=\bm{I}_{N}-\rho\bm{W}$.
However, since this is not a standard form, it is impossible to analytically derive the full conditional distribution.
Therefore, sampling is conducted using the Metropolis algorithm. If the value of $\rho$ in the $m$-th sampling is $\rho^{(m)}$, then the proposal value $\rho^{*}$ for the ($m+1$)-th sampling is determined by $\rho^{*}=\rho^{(m)}+ k \cdot N(0,1)$, where $k$ is a tuning parameter adjusted to achieve an acceptance rate of approximately $40\%$ to $60\%$.
The MCMC procedure for all parameters is detailed in the supplementary appendix.

Using the sampled parameters from the posterior distributions, we can estimate the treatment and spillover effects through the identification results (\ref{eq:identification_treatment_effect}) and (\ref{eq:identification_spillover_effect}) in Theorems \ref{Identification for Treatment Effect} and \ref{Identification for Spillover Effect}. Specifically, if $M$ samplings are performed, the estimated values of the treatment and spillover effects for the $m$-th sample are as follows:
\begin{align*}
  \xi_{0t}^{(m)}     & = Y_{0t} - \bm{\alpha}^{\top(m)}\big(\bm{I}_{N}-\rho^{(m)}\bm{w}\bm{\alpha}^{\top(m)}-\rho^{(m)}\bm{W}\big)^{-1}\big((\bm{I}_{N}-\rho^{(m)}\bm{W})\bm{Y}_{t}^{c}-\rho^{(m)}\bm{w}Y_{0t}\big), \\
  \bm{\xi}_{t}^{(m)} & = \bm{Y}_{t}^{c} - \big(\bm{I}_{N}-\rho^{(m)}\bm{w}\bm{\alpha}^{\top(m)}-\rho^{(m)}\bm{W}\big)^{-1}\big((\bm{I}_{N}-\rho^{(m)}\bm{W})\bm{Y}_{t}^{c}-\rho^{(m)}\bm{w}Y_{0t}\big).
\end{align*}
Calculating this for each $m\ (=1,2,\cdots,M)$ allows us to construct the distribution of the treatment and spillover effects.
Then the posterior means of $\xi_{0t}$ and $\bm{\xi}_{t}$ can be obtained by $(1/M)\sum_{m=1}^{M}\xi_{0t}^{(m)}$ and $(1/M)\sum_{m=1}^{M}\bm{\xi}_{t}^{(m)}$, respectively.
Given the estimators of $\bm{\alpha}$ and $\rho$, the estimation $\xi_{0t}$ and $\bm{\xi}_{t}$ depends only on the observed outcomes $Y_{0t}$ and $\bm{Y}_{t}^{c}$ over time. This suggests that time-varying components in the SAR panel data model, such as factors $\bm{\gamma}_{t}$, are not required to be estimated in the post-treatment periods $t\ (> T_{0})$.

\section{Simulation Study}\label{sec:simulation_study}


\subsection{Simulation Design}

We conduct a simulation study to examine the finite sample performance of the proposed Bayesian inference method. 
Three scenarios are considered for the number of untreated units $N$: $N=16$, $36$, and $64$. For the length of time periods, we consider two scenarios: $(T,T_{0})=(30,20)$ and $(60,50)$, where each scenario has a post-treatment length of ten.

In each scenario of $(N,T,T_0)$, we set up the DGPs as follows. 
We consider a network of $N$ untreated units represented by a rook matrix.\footnote{We use a rook matrix based on an $r$ board (so that $N=r^2$) to represent the network of $N$ untreated units. The rook matrix represents a square tessellation with connectivity of four for the inner fields on the chessboard, and two and three for the corner and border fields, respectively. 
}
Subsequently, we set the spatial weight $\omega_{ij}$ of any two untreated units to be $1$ if units $i$ and $j$ are connected in the network and $0$ otherwise.
Each element of the adjacency vector $\bm{w}$ takes the value of $1$ if $i \in \{1,2,3,4\}$ and $0$ otherwise.
The synthetic weights $\bm{\alpha}$ are set to provide large weights only to control units adjacent to the treatment unit, as follows:
\begin{align*}
  \alpha_{i} = \begin{cases}
                 0.5 & \mathrm{if}\ i=1, \\
                 -0.2 & \mathrm{if}\ i=2, \\
                 0.4 & \mathrm{if}\ i \in \{3,4\}, \\
                 0.1/6 & \mathrm{if}\ i \in \{5,6,\ldots,10\}, \\
                 0 & \mathrm{otherwise}.
               \end{cases}
\end{align*}

For each $t$ $(=1,2,\cdots,T)$, the control outcomes of untreated units $\bm{Y}_{t}^{c}(0)$ are distributed from a SAR panel data model as follows: $  \bm{Y}_{t}^{c}(0) = \rho\bm{w}Y_{0t}(0) + \rho\bm{W}\bm{Y}_{t}^{c}(0) + \bm{X}_{t}\bm{\beta} + \bm{u}_{t}$, where $\bm{X}_{t} = (X_{1t},\ldots,X_{NT})^{\top}$ with each element being i.i.d as $N(0,1)$, $\bm{u}_{t} = (u_{1t},\ldots,u_{NT})^{\top}$  with  each element being i.i.d as $N(0,1)$, and $\beta = 1.0$. The control outcome of the treated unit $Y_{0t}(0)$ is generated as $Y_{0t}(0) =\sum_{i=1}^{N}\alpha_{i}Y_{i,t}(0)$. We consider seven scenarios of $\rho$: $\rho=-0.8,\ -0.3,\ -0.1,\ 0.0,\ 0.1,\ 0.3,\ 0.8$. A larger absolute value of $\rho$ implies a stronger spatial correlation among the units. 

We generate the treatment outcomes of the treated unit $Y_{0t}(1)$ as $Y_{0t}(1) = Y_{0t}(0) + N(1, 1)$.
We also set
$  \bm{Y}_{t}^{c}(1) = \rho\bm{w}Y_{0t}(1) + \rho\bm{W}\bm{Y}_{t}^{c}(1) + \bm{X}_{t}\bm{\beta} + \bm{u}_{t}.$ The observed outcome for each $i$ and $t$ is $Y_{it} = Y_{it}(0)\cdot 1\{t\leq T_0\} + Y_{it}(1)\cdot 1\{t > T_0\}$. 

Following the DGPs, we conduct 1000 Monte Carlo simulations. For each simulation $r(=1,\cdots,1000)$, we draw $M(=5000)$ samples by MCMC and compute bias and root mean squared error (RMSE) of the posterior mean of treatment effect as follows:
\begin{align*}
 Bias &= \dfrac{1}{1000}\sum_{r=1}^{1000}\dfrac{1}{T-T_{0}}\sum_{t=T_{0}+1}^{T}\big(\xi_{0t}^{(r)}-\widehat{\xi}_{0t}^{(r)}\big), \\
  RMSE &= \sqrt{\dfrac{1}{1000}\sum_{r=1}^{1000}\dfrac{1}{T-T_{0}}\sum_{t=T_{0}+1}^{T}\big(\xi^{(r)}-\widehat{\xi}_{0t}^{(r)}\big)^{2}},
\end{align*}
where $\xi_{0t}^{(r)}$ denotes a true treatment effect in the $r$-th simulation, $\widehat{\xi}_{0t}^{(r)}=\sum_{m=1}^{M}\xi_{0t}^{(r,m)}/M$ is the estimated treatment effect with $\xi_{0t}^{(r,m)}$ being the estimate of the treatment effect in the $m$-th iteration for the $r$-th simulation. 

We compare the performance of the proposed method (labeled ``Proposed'') with those of the standard SCM  (labeled ``SCM'') \citep{abadie2010synthetic} and the Bayesian SCM of \cite{kim2020bayesian} (labeled ``BSCM'').\footnote{SCM does not involve MCMC sampling.}
We also calculate the 95\% coverage rate of treatment effect for the proposed method.


\subsection{Results}


Table \ref{5.2 Simulation} presents the simulation results for bias and RMSE. 
The key finding is that the error in estimating the treatment effect using SCM and BSCM becomes large as the absolute value of $\rho$ increases. 
Particularly in the case of a strong positive spatial correlation ($\rho=0.8$), SCM and BSCM exhibit substantial biases and large RMSEs. These results indicate that, when spillovers exist, the treatment effects estimated by SCM and BSCM are biased. Conversely, the proposed method exhibits a small bias and RMSE in each simulation scenario. The bias and RMSE of the proposed method do not increase with the magnitude of the spatial correlation. These results indicate that the proposed method is robust to the spillover effects arising from the spatial correlation of the outcomes.


Table \ref{5.2 Simulation coverage rate} presents 
the coverage rate of $95$\% credible interval of the treatment effect for the proposed method. In each scenario, the coverage rate is close to 95\%, indicating that the proposed inference method performs adequately. The inference performs well even when the length of pretreatment periods is not very large ($T_0 = 20$) and/or the spatial correlation is strong (i.e., $|\rho|$ is large).


\section{Empirical Application}
\label{sec:application}

We conduct two empirical studies applying the proposed method. The first estimates the impact of California’s tobacco tax on cigarette consumption \citep{abadie2010synthetic}, and the second examines the economic impact of Sudan’s 2011 division on GDP per capita. We present the results of both applications below.

\subsection{Application I: California Tobacco Tax}
\label{sec:application_tabacco_tax}

In this section, we apply the proposed method to estimate the effect of California tobacco tax (Proposition 99) on cigarette  consumption \citep{abadie2010synthetic}.
Proposition 99 is an anti-tobacco law issued in California in 1988 to promote awareness of the health risks associated with tobacco. It raised the excise tax on cigarettes by 25 cents per pack.
\citet{abadie2010synthetic} estimate the treatment effect of Proposition 99 on cigarette sales using the standard SCM, without accounting for spillover effects.
Our proposed method allows for spillovers and enables the estimation of both treatment and spillover effects of Proposition 99.


\subsubsection{Data}

We use annual state-level cigarette sales data in the U.S. from 1970 to 2000, as used in \cite{abadie2010synthetic}.\footnote{
This dataset is available from \citet{cunningham2021causal} on GitHub.  It has been preprocessed according to the procedures described in \citet{abadie2010synthetic}.
}
The outcome of interest is annual per capita cigarette consumption at the state level. 
We compare our proposed method with the standard SCM \citep{abadie2010synthetic} (labeled as ``SCM''), which is not robust to the existence of spillover effects.
For SCM, we use the synthetic weights estimated by \cite{abadie2010synthetic}.\footnote{The estimated values are listed in Table 2 of \citet{abadie2010synthetic}.}

In implementing the proposed method, we include the average retail cigarette prices for each year $t$ as covariates $\bm{X}_{t}$. For the spatial weights $w_{ij}$ in the SAR model (\ref{network model}), we construct a contiguity-based adjacency matrix using U.S. state boundary shapefiles from the Census TIGER/Line data. Specifically, we set $w_{ij}=1$ if states $i$ and $j$ share a common land border and $w_{ij}=0$ otherwise, and we set all diagonal elements $w_{ii}$ to zero. We then row-normalize $\bm{W}$ so that each row sums to one. We also construct the vector $\bm{w}$ by the same way.


\subsubsection{Results}



Figure \ref{fig:treatment_effect_tabacco} presents the estimation results for the synthetic control outcomes for California and the treatment effects of the California tobacco tax on cigarette consumption, with the posterior means reported for the proposed method.
Panel (a) shows that both the SCM and the proposed method fit well with per-capita cigarette sales in California during the pretreatment periods. 
Panel (b) indicates that each estimation method suggests Proposition 99 reduced cigarette consumption in California, with the proposed method exhibiting larger treatment effect estimates than SCM. The 90\% credible interval for the proposed method is negative for all post-treatment years, beginning in 1988. This result suggests that Proposition 99 negatively impacted cigarette sales for a decade, supporting the findings of \citet{abadie2010synthetic}. Figure \ref{fig:weights_tabacco} illustrates the posterior mean synthetic weights from our method along with the weights reported by \citet{abadie2010synthetic}.


Figure \ref{fig:spillover_tabacco} reports the estimated spillover effects for all states in the control group. The results align closely with economic and geographic intuition: the estimated spillover effects are largest in states geographically adjacent to California. The most pronounced negative effect appears in Nevada, a direct neighbor of California, while smaller yet persistent negative effects are also evident for Idaho and Utah.



\subsection{Application II: The Economic Cost of the 2011 Sudan Split}
\label{sec:application_sudan_split}

In this section, we assess the impact of Sudan's north-south split in 2011 on GDP per capita in the Sudans (the region of the former united Sudan) and other African countries.


\subsubsection{Background}

Sudan has long been divided along ethnic and religious lines, with Arabs (primarily Muslims) predominantly in the north and Africans (primarily Christians) in the south, leading to many conflicts.
Particularly, the Darfur conflict, driven by the Arab versus non-Arab ethnic divide, has persisted for many years in western Sudan.
This conflict has been marked by large-scale atrocities, including mass killings carried out by Arab militias known as ``Janjaweed.''

Amid these unending conflicts, the Sudanese government and the Sudan People's Liberation Army (SPLA), the main rebel force in Southern Sudan, signed the Comprehensive Peace Agreement (CPA) in 2005. This agreement, aimed at ending Sudan's civil wars, allowed South Sudan to establish its own government, achieve autonomy, and pursue independence through a referendum. Following the CPA, South Sudan voted for independence in January 2011 and was officially recognized as a nation on July 9, 2011.

Since 2011, South Sudan's independence has caused several economic disturbances in both Sudan and South Sudan. In particular, South Sudan's oil production shutdown in 2012 and its relapse into conflict in 2013 provoked a severe macroeconomic crisis in the region \citep{mawejje2020macroeconomic}. This study estimates the economic impact of Sudan's south-north split, with a focus on GDP per capita in the Sudans and other African countries.\footnote{Using the standard SCM, \citet{mawejje2021economic} estimate the economic losses in South Sudan, owing to the oil production halt in 2012, finding a nearly 70\% loss in per capita real GDP from 2012 to 2018. They excluded neighboring countries of South Sudan from their SCM analysis to avoid bias caused by spillovers; however, this practice can result in a poorly fitting synthetic control, making the perfect-fit assumption (Assumption \ref{Perfect fit}) less plausible and potentially causing additional bias.}


\subsubsection{Data}

We use data from African countries obtained from the World Bank DataBank. Our outcome of interest is ``GDP per capita (constant 2015 US\$)'' in the post-division period. 
The covariates we use include ``exports of goods and services (\% of GDP)'', ``merchandise trade (\% of GDP)'', ``access to electricity (\% of population)'',  ``inflation measured by the consumer price index (annual \%)'', ``net migration'', and ``trade (\% of GDP)''. 
Countries with missing values for the outcome or any covariates were excluded.

We focus on GDP per capita in the Sudans (i.e., the region of the former united Sudan) as the treated unit's outcome $Y_{0t}$.\footnote{The GDP per capita in the Sudans after the Sudan split is calculated by dividing the sum of GDP in North and South Sudan by the sum of their populations. Before the split, it corresponds to the GDP per capita in (the united) Sudan.}
The control group consists of $N = 29$ African countries with complete data from 2000 to 2015.\footnote{The control group includes: Algeria, Angola, Benin, Botswana, Burundi, Cameroon, Central African Republic, Chad, Egypt, Gabon, Ghana, Ivory Coast, Kenya, Madagascar, Mali, Mauritania, Mauritius, Morocco, Niger, Nigeria, Republic of the Congo, Rwanda, Senegal, South Africa, Tanzania, Togo, Tunisia, Uganda, and Zambia.}
Since South Sudan's independence occurred in July 2011, we define the pre-treatment periods as 2000–2010 and the post-treatment periods as 2011–2015. Although GDP per capita in the Sudans cannot be computed for 2011 due to incomplete data following the split, this does not affect the estimation of synthetic control outcomes in the pre-treatment periods or the treatment effects from 2012 onward.

Regarding the spatial weights in model (\ref{network model}), we specify $w_{ij}$ as the average bilateral trade volume between countries $i$ and $j$, normalized by the total trade of country $i$:
\begin{align*}
  w_{ij} &= \dfrac{\text{the average amount of trade between countries } i \text{ and } j}{\sum_{j}\text{the average amount of trade between countries } i \text{ and } j},
\end{align*}
where the trade data are obtained from the IMF and the averages are calculated over the pre-intervention periods.  Each weight $w_{ij}$ captures the strength of the economic ties between the two countries. Figure~\ref{fig:trade_Sudan} presents the resulting weights $w_{ij}$ between the former united Sudan and each control country during the pre-treatment period. Countries with stronger economic connections to the former unified Sudan are expected to experience more pronounced spillover effects from the Sudan split.


\subsubsection{Results}



Figure \ref{fig:treatment_effect_Sudan} presents the estimation results for the synthetic outcomes of the Sudans and the treatment effects of the Sudan split.
Across all methods, the estimates indicate a negative impact of South Sudan’s independence on GDP per capita in the Sudans. In particular, the proposed method estimates a decline of approximately 100USD in 2012, corresponding to about 7.8\% of the actual GDP per capita for that year.
In 2015, the estimated GDP per capita in the synthetic Sudan is about $9.5\%$ higher than the actual GDP per capita.
Furthermore, the cumulative losses in the Sudans are estimated to have reached $34\%$ from 2011 to 2015.\footnote{The losses in the Sudans are defined by $100\times(\xi_{0t}/Y_{0t}(0))$ (\%) for each $t>T_{0}$ and those in control countries are defined by $100\times\xi_{it}/Y_{it}(0)$ (\%) for each $i=1,2,\cdots,N$ and $t>T_{0}$.
The cumulative losses are computed by $100\times \sum_{t=2011}^{2015}\xi_{it}/Y_{it}(0)$ for each $i=0,1,\cdots,29$. We estimate these by the posterior means of these losses.} Figure \ref{fig:weights_Sudan} displays the posterior mean weights from our Bayesian SCM alongside the weights estimated by the standard SCM.


Figure \ref{fig:spillover_Sudan} presents the estimated spillover effects of the Sudan split on other African countries. Countries with substantial trade volumes with the former united Sudan, such as Egypt and Kenya, experienced pronounced negative spillover effects from the split.

The political instability and north–south split in Sudan significantly altered the country’s industrial structure, which in turn had a profound impact on trade. For example, although South Sudan is rich in oil and other natural resources, the conflict between the North and South led to a halt in oil production in 2012, resulting in the loss of critical export commodities. This disruption also affected countries with close economic ties to Sudan through trade channels.
Thus, the north–south split not only had adverse economic consequences for the Sudans themselves but also generated negative spillover effects on other countries.
This empirical study illustrates that political and economic changes in one country can have broader consequences for other nations with strong economic linkages.


\section{Conclusion}\label{sec:conclusion}

This study extends the SCM to account for spillover effects. Although SCM is frequently applied to spatial data where spillovers may be present, the conventional approach relies on SUTVA, which can yield biased estimates when spillovers are present.
To address this limitation, we propose a novel SCM that incorporates the SAR panel data model to capture spillover effects. We also develop a Bayesian inference procedure that estimates both treatment and spillover effects, using horseshoe priors for regularization.
We apply the method to two empirical studies:
(i) evaluating the impact of the California tobacco tax on cigarette consumption \citep{abadie2010synthetic}, and (ii) assessing the economic impact of the 2011 Sudan division on GDP per capita. The first application reveals a negative impact of the tax on cigarette consumption in California and other US states. The second shows that the Sudan split substantially reduced GDP per capita in the Sudans and caused negative spillover effects
on other African countries with strong economic ties to the former united Sudan.


\bigskip

\subsection*{Acknowledgments}
We thank the editor and two anonymous referees for helpful comments that greatly
improved the quality of the paper. We also thank Kaoru Irie, Ryo Okui, Yasuyuki Sawada, and participants in various seminars and workshops for valuable comments. The authors gratefully acknowledge the financial support from JSPS KAKENHI Grant (number 24K16342).



\setstretch{1.5} 

\bibliography{ref}


\newpage
\section*{Tables}
\begin{table}[H]
\centering
\caption{Simulation Results for Bias and RMSE}
\label{5.2 Simulation}
\caption*{Panel (a): $T_0=20$}
\label{5.2 Simulation panel a}
\begin{tabular}{llrrrrrrrrr}
\hline
\multicolumn{2}{l}{} & 
\multicolumn{3}{c}{$N=16$} & \multicolumn{3}{c}{$N=36$} & \multicolumn{3}{c}{$N=64$}
 \\
\cmidrule(lr){3-5}\cmidrule(lr){6-8}\cmidrule(lr){9-11}

\multicolumn{1}{l}{} & \multicolumn{1}{l}{$\rho$} & 
Proposed & SCM & BSCM & Proposed & SCM & BSCM & Proposed & SCM & BSCM
 \\ \hline
 & -0.8 & -0.001 & -0.140 & -0.172 & -0.001 & -0.135 & -0.180 & -0.000 & -0.113 & -0.176 \\
& -0.3 & -0.001 & -0.063 & -0.071 & -0.001 & -0.060 & -0.073 & -0.001 & -0.056 & -0.072 \\
& -0.1 & -0.003 & -0.025 & -0.026 & -0.002 & -0.022 & -0.027 & -0.001 & -0.023 & -0.027 \\
Bias & 0 & -0.002 & -0.003 & 0.000 & -0.000 & 0.004 & 0.001 & 0.000 & -0.000 & 0.001 \\
& 0.1 & -0.002 & 0.022 & 0.029 & -0.002 & 0.028 & 0.028 & -0.001 & 0.022 & 0.029 \\
& 0.3 & -0.003 & 0.092 & 0.101 & -0.002 & 0.087 & 0.100 & 0.000 & 0.085 & 0.098 \\
& 0.8 & 0.001 & 0.504 & 0.569 & -0.008 & 0.456 & 0.515 & -0.005 & 0.466 & 0.512 \\
\hdashline
 & -0.8 & 0.008 & 0.467 & 0.245 & 0.030 & 0.536 & 0.257 & 0.036 & 0.583 & 0.255 \\
& -0.3 & 0.017 & 0.354 & 0.101 & 0.032 & 0.385 & 0.108 & 0.037 & 0.427 & 0.110 \\
& -0.1 & 0.024 & 0.323 & 0.037 & 0.033 & 0.345 & 0.047 & 0.043 & 0.384 & 0.057 \\
RMSE & 0 & 0.026 & 0.316 & 0.000 & 0.034 & 0.329 & 0.029 & 0.041 & 0.359 & 0.039 \\
& 0.1 & 0.029 & 0.308 & 0.042 & 0.035 & 0.326 & 0.050 & 0.044 & 0.352 & 0.057 \\
& 0.3 & 0.040 & 0.332 & 0.144 & 0.041 & 0.335 & 0.143 & 0.050 & 0.345 & 0.144 \\
& 0.8 & 0.126 & 0.956 & 0.807 & 0.116 & 0.840 & 0.733 & 0.123 & 0.812 & 0.729 \\
\hline
\end{tabular}

\bigskip

\caption*{Panel (b): $T_0=50$}
\label{5.2 Simulation panel b}
\begin{tabular}{llrrrrrrrrr}
\hline
\multicolumn{2}{l}{} & 
\multicolumn{3}{c}{$N=16$} & \multicolumn{3}{c}{$N=36$} & \multicolumn{3}{c}{$N=64$}
 \\
\cmidrule(lr){3-5}\cmidrule(lr){6-8}\cmidrule(lr){9-11}

\multicolumn{1}{l}{} & \multicolumn{1}{l}{$\rho$} & 
Proposed & SCM & BSCM & Proposed & SCM & BSCM & Proposed & SCM & BSCM
 \\ \hline
 & -0.8 & -0.000 & -0.144 & -0.175 & -0.000 & -0.147 & -0.182 & -0.000 & -0.148 & -0.182 \\
& -0.3 & -0.001 & -0.069 & -0.072 & -0.000 & -0.065 & -0.071 & -0.000 & -0.061 & -0.075 \\
& -0.1 & -0.000 & -0.028 & -0.026 & -0.001 & -0.024 & -0.026 & -0.000 & -0.023 & -0.026 \\
Bias & 0 & -0.001 & -0.001 & -0.000 & 0.000 & -0.004 & -0.000 & -0.000 & 0.001 & 0.000 \\
& 0.1 & -0.000 & 0.021 & 0.029 & -0.001 & 0.028 & 0.029 & -0.000 & 0.026 & 0.029 \\
& 0.3 & -0.001 & 0.088 & 0.101 & -0.001 & 0.084 & 0.098 & -0.000 & 0.086 & 0.097 \\
& 0.8 & 0.001 & 0.502 & 0.573 & -0.002 & 0.466 & 0.518 & -0.001 & 0.457 & 0.503 \\
\hdashline
 & -0.8 & 0.005 & 0.439 & 0.247 & 0.004 & 0.465 & 0.256 & 0.003 & 0.476 & 0.258 \\
& -0.3 & 0.011 & 0.327 & 0.102 & 0.008 & 0.327 & 0.102 & 0.006 & 0.333 & 0.105 \\
& -0.1 & 0.014 & 0.300 & 0.037 & 0.010 & 0.300 & 0.037 & 0.008 & 0.302 & 0.037 \\
RMSE & 0 & 0.016 & 0.286 & 0.000 & 0.012 & 0.288 & 0.000 & 0.009 & 0.289 & 0.000 \\
& 0.1 & 0.018 & 0.286 & 0.041 & 0.013 & 0.286 & 0.041 & 0.010 & 0.290 & 0.041 \\
& 0.3 & 0.025 & 0.307 & 0.142 & 0.017 & 0.301 & 0.140 & 0.013 & 0.303 & 0.138 \\
& 0.8 & 0.078 & 0.946 & 0.812 & 0.059 & 0.834 & 0.730 & 0.050 & 0.797 & 0.720 \\
\hline
\end{tabular}

\medskip
\begin{tablenotes}
{\footnotesize
\item Notes: Panels (a) and (b) show the simulation results for the bias and RMSE for $T_0=20$ and $50$, respectively. Each panel shows the simulation results for each of the proposed method, SCM, and BSCM, and each of $\rho \in \{-0.8,-0.3,-0.1,0.0,0.1,0.3,0.8\} $ and $N \in \{16,36,64\}$.
}
\end{tablenotes}
\end{table}

\begin{table}[H]
\centering
\caption{Coverage Rate of 95\% Credible Interval for the Proposed Method}
\label{5.2 Simulation coverage rate}
\begin{tabular}{lrrrrrr}
\hline
$\rho$ & \multicolumn{3}{c}{$T_{0}=20$} & \multicolumn{3}{c}{$T_{0}=50$} \\
\cmidrule(lr){2-4} \cmidrule(lr){5-7}
$\rho$ & $N=16$ & $N=36$ & $N=64$ & $N=16$ & $N=36$ & $N=64$ \\ \hline
-0.8 & 0.938 & 0.956 & 0.954 & 0.952 & 0.951 & 0.942 \\
-0.3 & 0.946 & 0.976 & 0.975 & 0.955 & 0.938 & 0.961 \\
-0.1 & 0.946 & 0.982 & 0.978 & 0.948 & 0.955 & 0.951 \\
 0.0 & 0.958 & 0.985 & 0.981 & 0.952 & 0.951 & 0.944 \\
 0.1 & 0.951 & 0.984 & 0.984 & 0.954 & 0.955 & 0.944 \\
 0.3 & 0.942 & 0.984 & 0.984 & 0.944 & 0.947 & 0.940 \\
 0.8 & 0.953 & 0.985 & 0.989 & 0.954 & 0.948 & 0.943 \\
\hline
\end{tabular}
\begin{tablenotes}
{\footnotesize
\item Notes: This table shows the coverage rate of $95\%$ credible interval for the proposed method for each scenario of $\rho$, $T_0$, and $N$. The coverage rate is computed over 1000 simulations per scenario.
}
\end{tablenotes}
\end{table}


\newpage
\newgeometry{left=1cm, right=1cm, top=2cm, bottom=2cm}
\section*{Figures}
\begin{figure}[H]
\caption{Estimates of the Counterfactual Outcomes  and Treatment Effects of California Tobacco Tax}
\label{fig:treatment_effect_tabacco}
  \begin{tabular}{cc}
    \begin{minipage}[t]{0.5\hsize}
    \subcaption{Observed Outcomes and Synthetic Control Outcomes}
      \centering
      \includegraphics[keepaspectratio, scale=0.4]{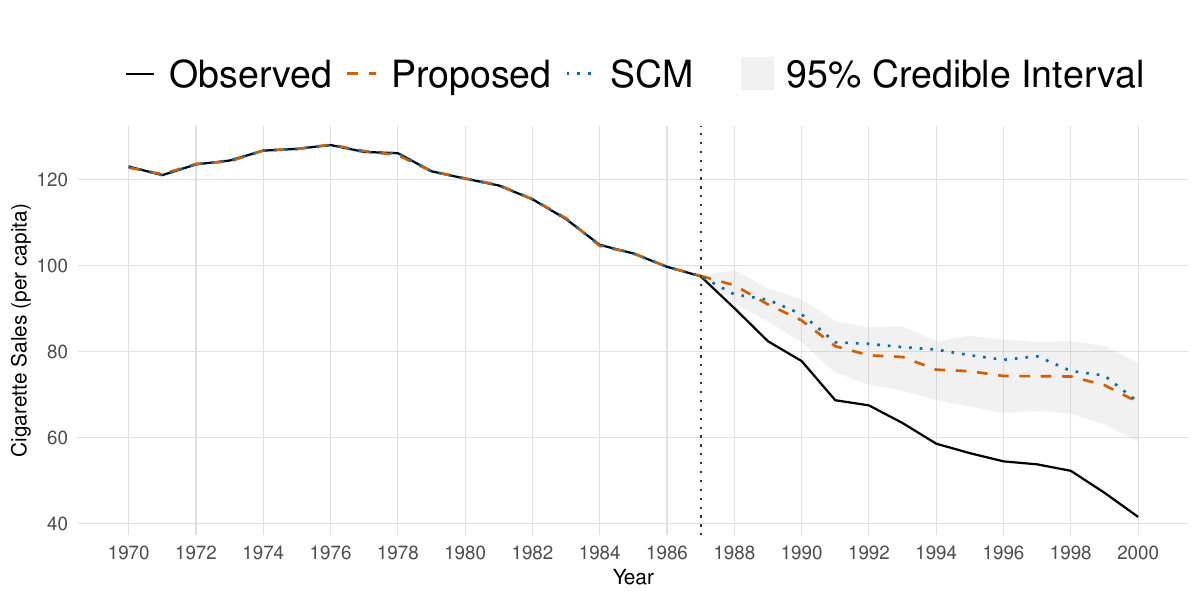}
    \end{minipage} &
    \begin{minipage}[t]{0.5\hsize}
    \subcaption{Treatment Effect Estimation}
      \centering
      \includegraphics[keepaspectratio, scale=0.4]{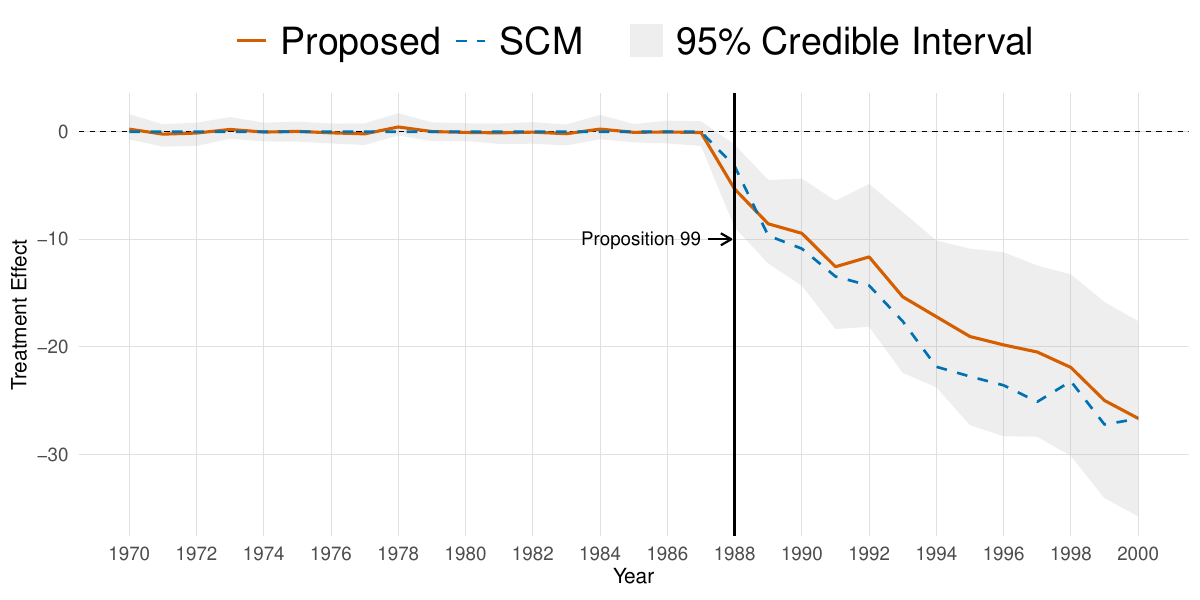}
    \end{minipage}
  \end{tabular}
  \begin{tablenotes}
    {\footnotesize
    \item Note: Panel (a) shows the observed outcomes (per-capita cigarette sales in packs) for California (black dotted line) and estimates of the counterfactual control outcomes for California using the proposed method (red solid line) and SCM (blue dashed line). Panel (b) shows the estimates of the treatment effects of Proposition 99 using the proposed method (red solid line) and the SCM (black dashed line). 
    Each panel also illustrates a 95\% credible interval by the proposed method.
    Post-mean of $\rho$ is $0.185$ with a $95\%$ posterior credible interval of $[0.096 ,0.268 ]$.}
  \end{tablenotes}
\end{figure}

\bigskip
\bigskip

\begin{figure}[H]
\caption{Estimated Weights for California’s Synthetic Control}
\label{fig:weights_tabacco}
  \centering
  \includegraphics[keepaspectratio, scale=0.6]{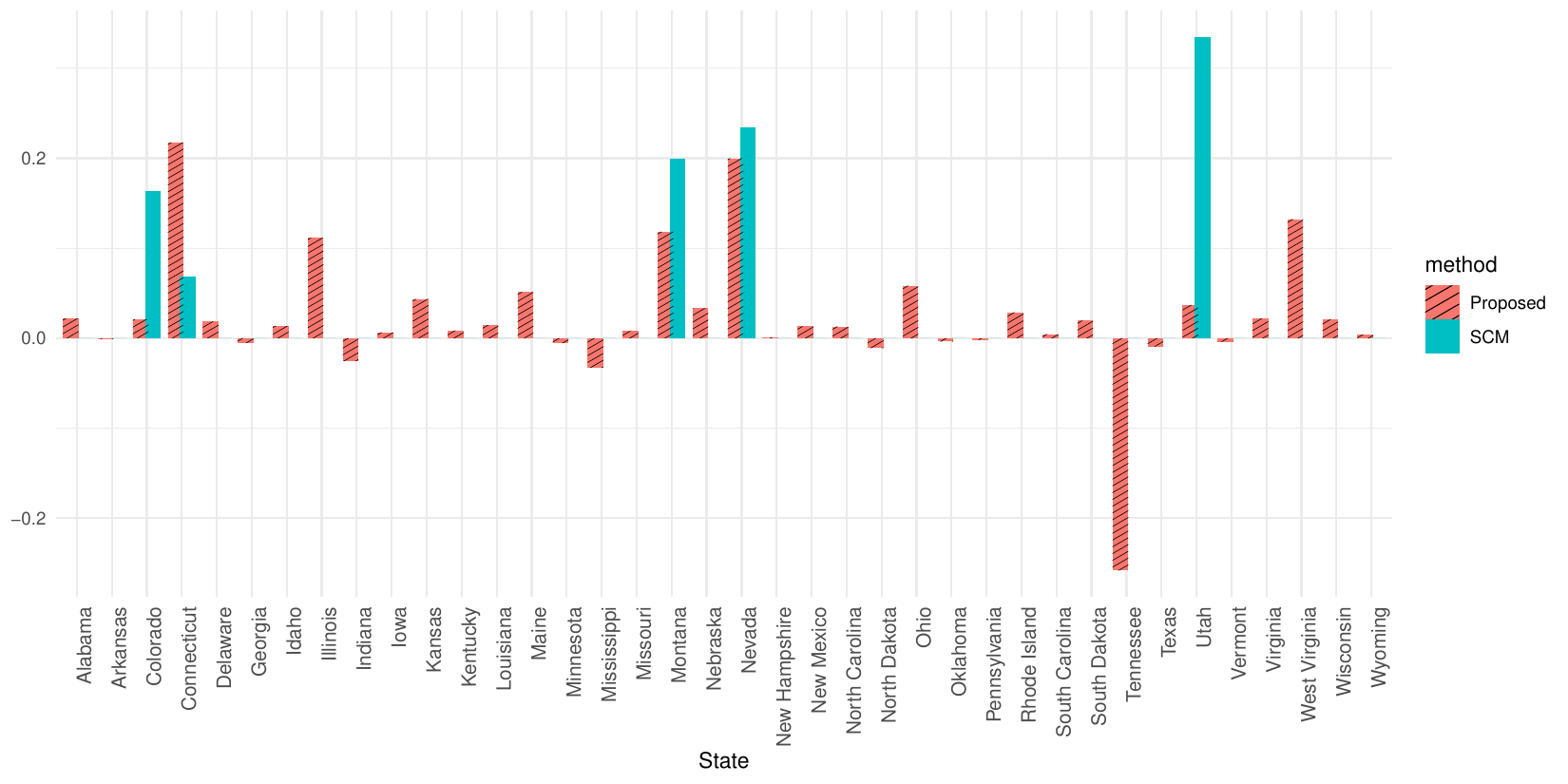}
  \begin{tablenotes}
    {\footnotesize
    \item Note: This bar chart displays the posterior mean synthetic weights $\bm{\alpha}$ from our proposed method (red bars) as well as the weights obtained by \citet{abadie2010synthetic} via the original SCM (blue bars). Each bar represents the contribution of a given donor unit to California’s synthetic control outcomes.}
  \end{tablenotes}
\end{figure}

\bigskip
\bigskip

\begin{figure}[H]
  \centering
    \caption{Estimates of the Spillover Effects of California Tobacco Tax}
    \label{fig:spillover_tabacco}
  \includegraphics[scale=0.8]{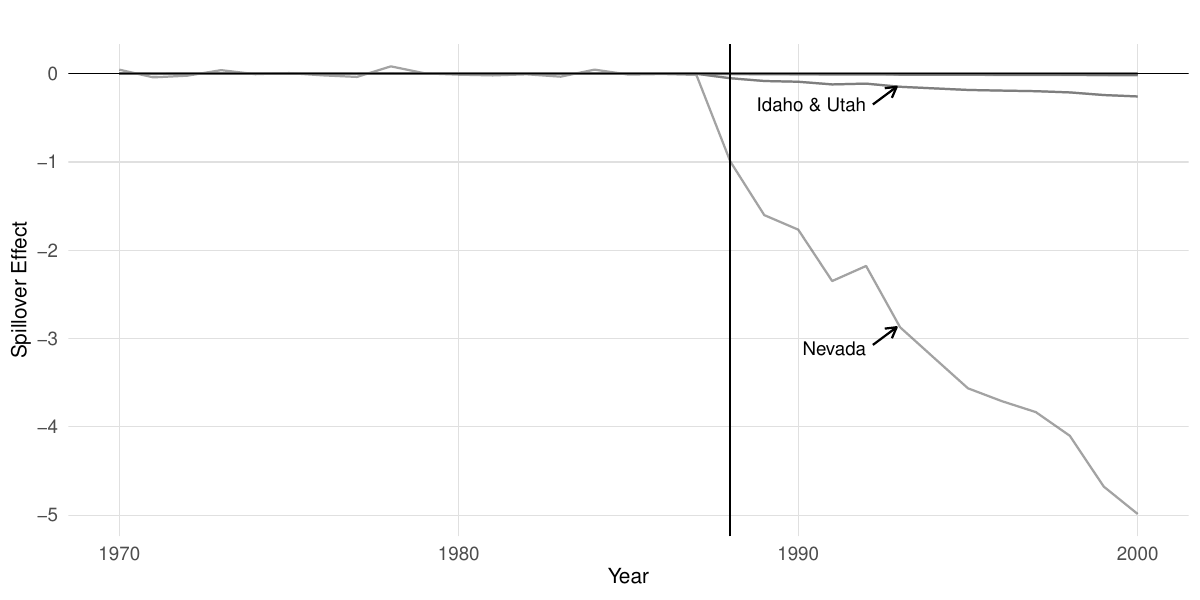}
  \begin{tablenotes}
    {\footnotesize
    \item Notes: This figure shows the spillover effects of Proposition 99 estimated by the proposed method for each state in the control group. For ease of reading, only three states with particularly strong spillover effects are indicated.}
  \end{tablenotes}
\end{figure}

\bigskip

\begin{figure}[H]
  \centering
  \caption{Bilateral Trade Volume between Sudan and Control Countries}
  \label{fig:trade_Sudan}
  \includegraphics[scale=0.8]{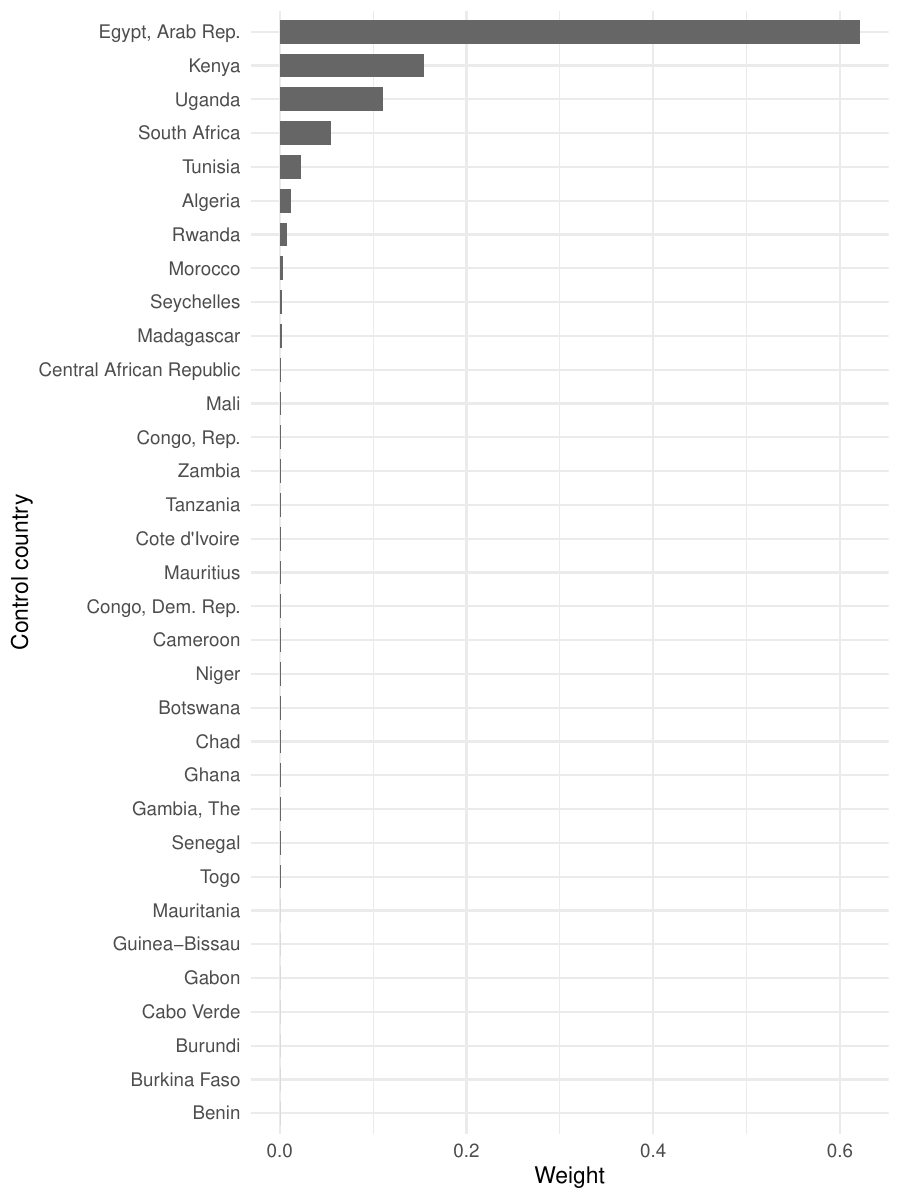}
  \begin{tablenotes}
   {\footnotesize
    \item Notes: This figure displays the spatial weights ($w_i$) used in the Sudan application, which are constructed based on the average bilateral trade volume between the former united Sudan and each control country during the pre-treatment period. A higher value indicates a stronger economic linkage.}
  \end{tablenotes}
\end{figure}

\bigskip

\begin{figure}[H]
  \centering
  \caption{Estimates of the Spillover Effects of Sudan Split}
  \label{fig:spillover_Sudan}
  \includegraphics[scale=0.8]{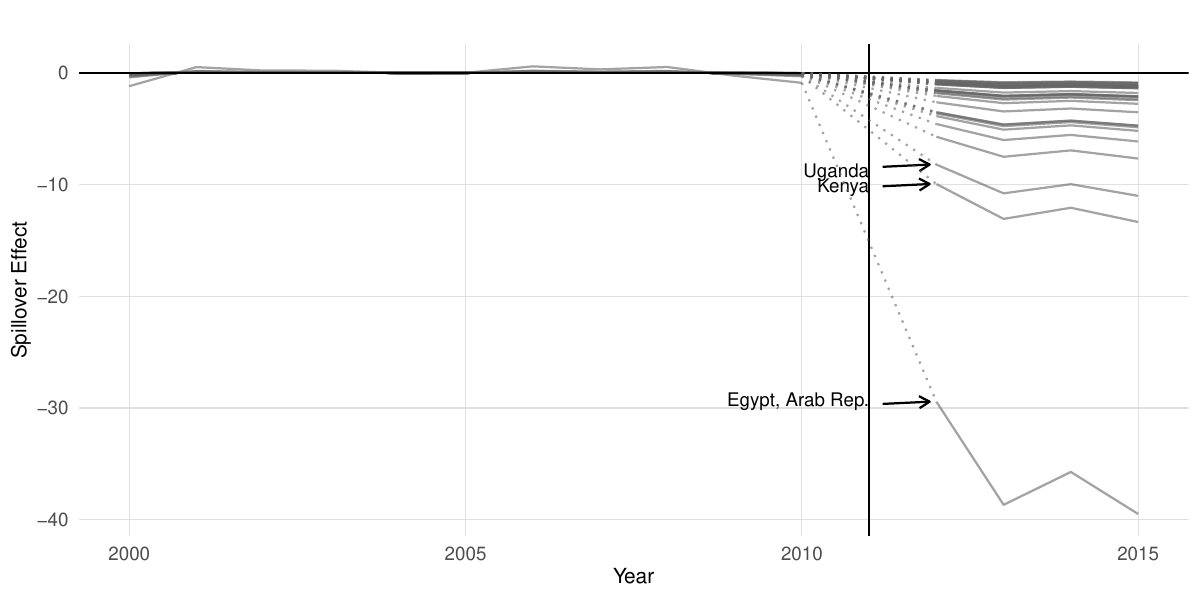}
  \begin{tablenotes}
   {\footnotesize
    \item Notes: This figure shows the spillover effects of the Sudan split, estimated by the proposed method for each country in the control group. Dotted lines illustrate missing values. To ensure readability, only four countries with strong spillover effects are indicated.}
  \end{tablenotes}
  \label{7.3 spillover}
\end{figure}

\bigskip

\begin{figure}[H]
  \caption{Estimates of the Counterfactual Outcomes and Treatment Effects of Sudan Split}
  \label{fig:treatment_effect_Sudan}
  \begin{tabular}{cc}
    \begin{minipage}[t]{0.5\hsize}
          \subcaption{Observed Outcomes and Synthetic Control Outcomes}
      \centering
      \includegraphics[keepaspectratio, scale=0.45]{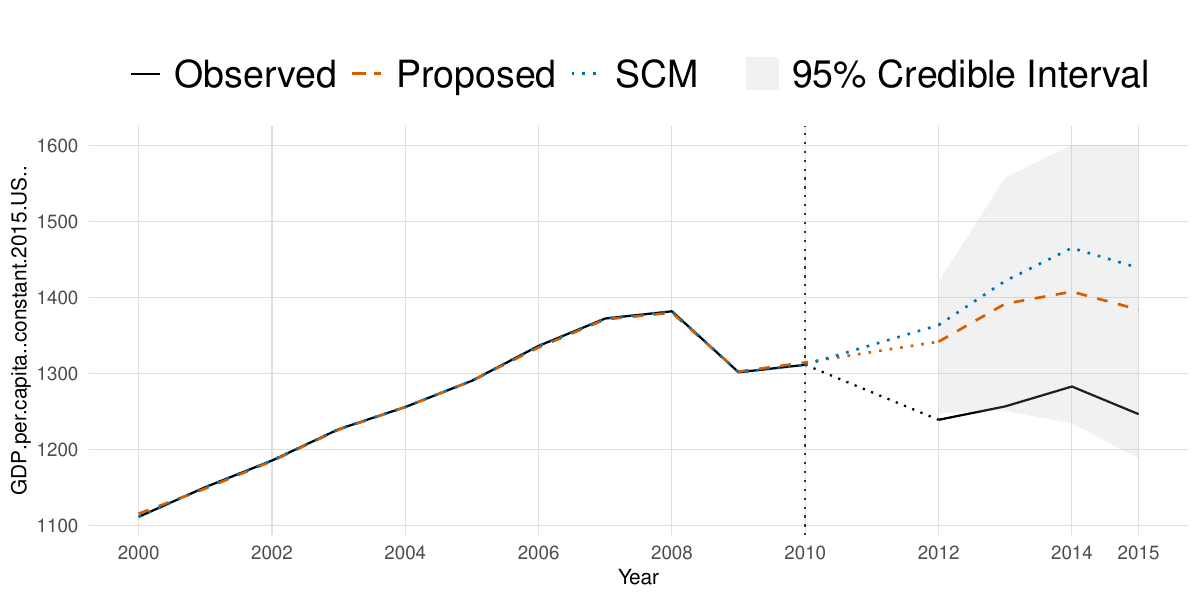}
    \end{minipage} &
    \begin{minipage}[t]{0.5\hsize}
    \subcaption{Estimates of Treatment Effects}
      \centering
      \includegraphics[keepaspectratio, scale=0.45]{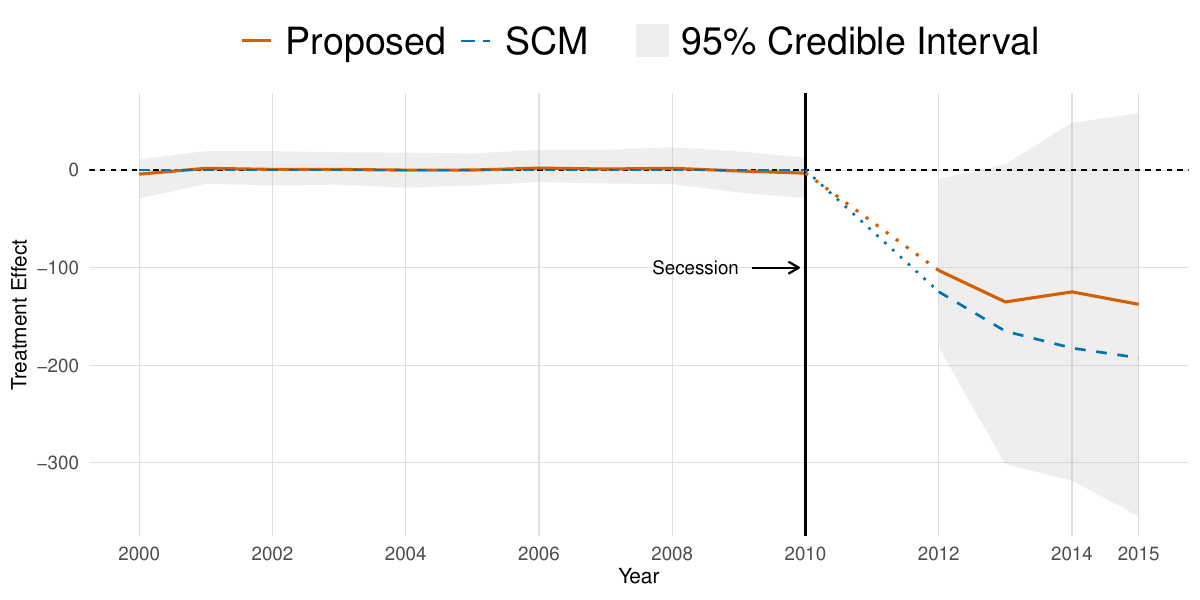}
    \end{minipage}
  \end{tabular}  
  \begin{tablenotes}
      {\footnotesize
    \item Notes: Panel (a) shows the observed outcomes (GDP per capita) of the Sudans (black dotted-dashed line) and estimates of the counterfactual control outcomes for the Sudans using the proposed method (red solid line) and the SCM (blue dashed line). Panel (b) shows the estimates of the treatment effects of the Sudan split using the proposed method (red solid line) and the SCM (blue dashed line). In each panel, missing values are illustrated with dotted lines.
    Each panel also illustrates a 95\% credible interval by the proposed method.
    Post-mean of $\rho$ is about $0.427 $ ($95\%$ posterior credible interval is $[0.294 ,0.550]$).}
  \end{tablenotes}
\end{figure}

\bigskip
\bigskip

\begin{figure}[H]
\caption{Estimated Weights in the Synthetic Control Construction for the Sudans}
\label{fig:weights_Sudan}
  \centering
  \includegraphics[keepaspectratio, scale=0.6]{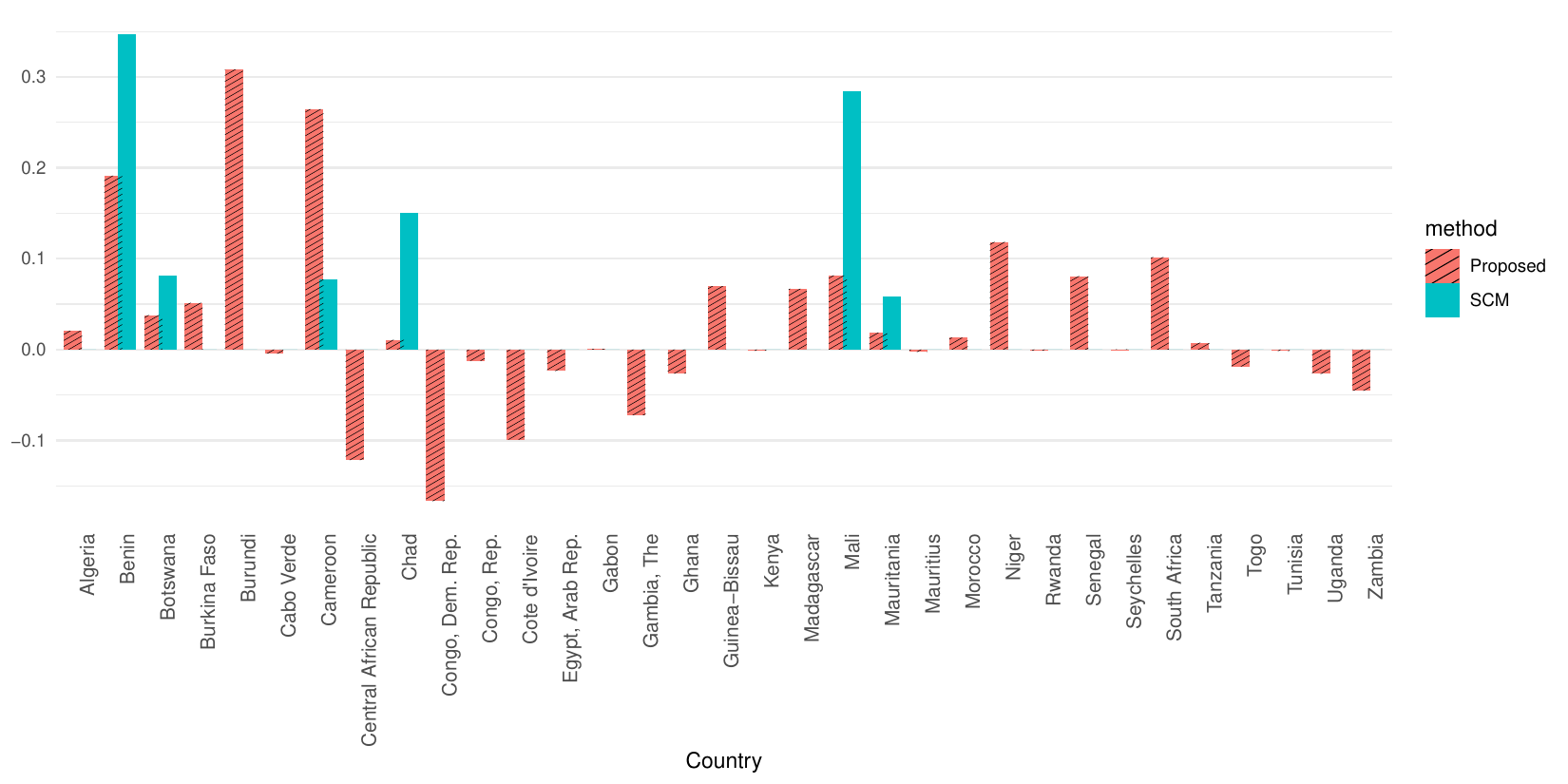}
  \begin{tablenotes}
    {\footnotesize
    \item Note: This bar chart displays the posterior mean weights from our proposed Bayesian SCM (red bars) as well as the weights estimated by the original SCM of \citet{abadie2010synthetic} (blue bars). Each bar indicates the contribution of a given donor country to the counterfactual GDP per capita trajectory for the Sudans.}
  \end{tablenotes}
\end{figure}

\restoregeometry


\newpage
\appendix

\appendix
\part*{Appendix}
\label{appendix}

To derive full conditional distributions, we use the following proposition:
\begin{prop}[\citet{10.1214/11-BA631}]
  \begin{align*}
    X\mid a \sim \text{Half-Cauchy}(0,a) \Longleftrightarrow \begin{cases}
                                                               X^{2}\mid b \sim \text{Inverse-Gamma}\qty(\dfrac{1}{2},\dfrac{1}{b}) \\
                                                               b\mid a \sim \text{Inverse-Gamma}\qty(\dfrac{1}{2},\dfrac{1}{a^{2}})
                                                             \end{cases}
  \end{align*}
\end{prop}
This proposition implies that a half-Cauchy distribution is equivalent to a hierarchical form of an inverse gamma distribution with an auxiliary variable.

\section{Synthetic Weights}
By using auxiliary variables, we can derive the full conditional distributions as follows:
\begin{align*}
  \bm{\alpha} \mid \text{rest}    & \sim \mathcal{N}_{N}\big(A^{-1}(\bm{Y}^{c}(0))^{\top}\bm{Y}_{0}(0),\ \sigma^{2}_{1}\bm{A}^{-1}\big),                                                                                                                    \\
  \text{where}\ \ \bm{A}          & = (\bm{Y}^{c}(0))^{\top}\bm{Y}^{c}(0) + \sigma^{2}_{1}\mathrm{diag}(1/\lambda^{2}_{1},1/\lambda^{2}_{2},\cdots,1/\lambda^{2}_{N})\in\mathbb{R}^{N\times N}                                                              \\
  \lambda^{2}_{i}\mid\text{rest}  & \sim \text{Inverse-Gamma}\qty(1,\ \dfrac{\alpha_{i}^{2}}{2} + \dfrac{1}{\nu_{\lambda_{i}}}),\ \ \text{for}\ i=1,2,\cdots,N                                                                                                  \\
  \nu_{\lambda_{i}}\mid\text{rest}    & \sim \text{Inverse-Gamma}\qty(1,\ \dfrac{1}{\lambda^{2}_{i}}+\dfrac{1}{\tau^{2}})                                                                                                                         \\
  \tau^{2}\mid\text{rest}         & \sim \text{Inverse-Gamma}\qty(\dfrac{N+1}{2},\ \sum_{i=1}^{N}\dfrac{1}{\nu_{\lambda_{i}}} + \dfrac{1}{\nu_{\tau}})                                                                                                                                     \\
  \nu_{\tau}\mid\text{rest}       & \sim \text{Inverse-Gamma}\qty(1,\ \dfrac{1}{\tau^{2}}+\dfrac{1}{\sigma_{1}^{2}})                                                                                                                                        \\
  \sigma^{2}_{1}\mid\text{rest}   & \sim \text{Inverse-Gamma}\qty(1+\dfrac{T_{0}}{2},\ \dfrac{1}{\nu_{\tau}}+\dfrac{1}{\nu_{\sigma_{1}}}+\dfrac{1}{2}(\bm{Y}_{0}(0)-\bm{\alpha}^{\top}\bm{Y}^{c}(0))^{\top}(\bm{Y}_{0}(0)-\bm{\alpha}^{\top}\bm{Y}^{c}(0))) \\
  \nu_{\sigma_{1}}\mid\text{rest} & \sim \text{Inverse-Gamma}\qty(1,\ \dfrac{1}{\sigma_{1}^{2}}+\dfrac{1}{10^{2}})
\end{align*}

\section{Spatial Autoregressive Panel Data Model}
For the remaining parameters, we specify priors as follows.
\begin{enumerate}
  \item Full conditionals of $\bm{\beta}$:
  \begin{align*}
          \bm{\beta} \mid \text{rest}            & \sim \mathcal{N}_{k}(\bm{A}_{\beta}^{-1}\tilde{u},\ \sigma^{2}_{2}\bm{A}_{\beta}^{-1})                                                   \\
          \text{where}\ \ \bm{A}_{\beta}         & =\bm{X}^{\top}\bm{X} + \sigma^{2}_{2}\mathrm{diag}(1/\lambda_{\beta_{1}}^{2},1/\lambda_{\beta_{2}}^{2},\cdots,1/\lambda_{\beta_{k}}^{2}) \\
          \lambda_{\beta_{j}}^{2}\mid\text{rest} & \sim \text{Inverse-Gamma}\qty(1,\ \dfrac{\beta_{j}^{2}}{2}+\dfrac{1}{\nu_{\lambda_{\beta_{0}}}}),\ \ \text{for}\ j=1,2,\cdots,k          \\
          \nu_{\lambda_{\beta}}\mid\text{rest}   & \sim \text{Inverse-Gamma}\qty(1,\ \dfrac{1}{\tau^{2}_{\beta}} + \sum_{j=1}^{k}\dfrac{1}{\lambda^{2}_{\beta_{j}}})                        \\
          \tau^{2}_{\beta} \mid\text{rest}       & \sim \text{Inverse-Gamma}\qty(1,\ \dfrac{1}{\nu_{\lambda_{\beta}}} + \dfrac{1}{\nu_{\tau_{\beta}}})                                      \\
          \nu_{\tau_{\beta}}\mid\text{rest}      & \sim \text{Inverse-Gamma}\qty(1,\ \dfrac{1}{\tau^{2}_{\beta}} + \dfrac{1}{\sigma_{2}^{2}});
        \end{align*}
  \item Full conditionals of $\phi_{\gamma}$ and $\sigma^{2}_{\gamma}$:
  \begin{align*}
          \phi_{\gamma}\mid\text{rest}       & \sim N\qty(\dfrac{\sum_{t=1}^{T_{0}}\bm{\gamma}_{t-1}^{\top}\bm{\gamma}_{t}}{\sum_{t=1}^{T_{0}}\bm{\gamma}_{t-1}^{\top}\bm{\gamma}_{t-1}},\ \dfrac{\sigma^{2}_{\gamma}}{\sum_{t=1}^{T_{0}}\bm{\gamma}_{t-1}^{\top}\bm{\gamma}_{t-1}}),   \\
          \sigma^{2}_{\gamma}\mid\text{rest} & \sim \text{Inverse-Gamma}\qty(\dfrac{1}{2}+\dfrac{pT_{0}}{2},\ \dfrac{1}{\nu_{\sigma_{\gamma}}} + \dfrac{1}{2}\sum_{t=1}^{T_{0}}(\bm{\gamma}_{t}-\phi_{\gamma}\bm{\gamma}_{t-1})^{\top}(\bm{\gamma}_{t}-\phi_{\gamma}\bm{\gamma}_{t-1})), \\
          \nu_{\sigma}\mid\text{rest}        & \sim \text{Inverse-Gamma}\qty(1,\ \dfrac{1}{\sigma^{2}_{\gamma}}+\dfrac{1}{10^{2}}).
        \end{align*}
        After drawing samples from the above full conditional distributions, we update $\bm{\gamma}_{t}$ using the forward filtering backward sampling method.
  \item Full conditionals of $\bm{\eta}$, $\sigma^{2}_{\eta}$ and $\omega_{1}^{2},\cdots,\omega_{p}^{2}$:
  \begin{align*}
          \underset{p\times 1}{\bm{\eta}_{i}} \mid\text{rest} & \sim \mathcal{N}_{p}\qty(\bm{C}^{-1}\sum_{t=1}^{T_{0}}\bm{\gamma}_{t}(\bm{Y}^{c}(0)-u_{it}),\ \sigma^{2}_{2}\bm{C}^{-1}),\ \ \text{for}\ i=1,2,\cdots,N                                   \\
          \text{where}\ \ \bm{C}                              & = \sum_{t=1}^{T_{0}}\bm{\gamma}_{t}\bm{\gamma}_{t}^{\top} + \sigma^{2}_{2}\sigma^{-2}_{\eta}\bm{\Sigma}_{\eta}^{-1},\ \ \bm{\Sigma}_{\eta}=\mathrm{diag}(\omega_{1}^{2},\cdots,\omega_{p}^{2}),                                                                           \\
          \text{and}\ \ u_{it}                                & = \rho w_{i} Y_{0t}(0)+\rho\sum_{j=1}^{N}W_{ij}Y_{jt}(0) + \bm{X}_{it}\bm{\beta}_{0}                                                                                                             \\
          \sigma^{2}_{\eta}\mid\text{rest}                    & \sim \text{Inverse-Gamma}\qty(\dfrac{1}{2}+\dfrac{pN}{2},\ \dfrac{1}{\nu_{\sigma_{\eta}}}+\dfrac{1}{2}\sum_{i=1}^{N}\bm{\eta}_{i}^{\top}\bm{\Sigma}_{\eta}^{-1}\bm{\eta}_{i})               \\
          \omega_{j}^{2}\mid\text{rest}                       & \sim \text{Inverse-Gamma}\qty(\dfrac{1}{2}+\dfrac{N}{2},\ \dfrac{1}{\nu_{\omega_{j}}} + \dfrac{1}{2}\sum_{i=1}^{N}\dfrac{(\eta_{ij})^{2}}{\sigma^{2}_{\eta}}),\ \ \text{for}\ j=1,2,\cdots,p \\
          \nu_{\omega_{j}}\mid\text{rest}                     & \sim \text{Inverse-Gamma}\qty(1,\ \dfrac{1}{\omega_{j}^{2}}+\dfrac{1}{10^{2}}),\ \ \text{for}\ j=1,2,\cdots,p.
        \end{align*}
  \item Full conditionals of $\sigma^{2}_{2}$:
  \begin{align*}
          \sigma^{2}_{2}\mid\text{rest}   & \sim \text{Inverse-Gamma}\qty(1 + \dfrac{NT_{0}}{2},\ \dfrac{1}{\nu_{\beta}} + \dfrac{1}{\nu_{\sigma_{2}}} +\sum_{t=1}^{T_{0}}\bm{u}_{t}^{\top}\bm{u}_{t}) \\
          \text{where}\ \ \bm{u}_{t}      & = \bm{Y}_{t}^{c}(0)-\rho\bm{w}Y_{0t}(0)-\rho\bm{W}\bm{Y}_{t}^{c}(0)-\bm{X}_{t}\bm{\beta}_{0}-\bm{\eta}^{0}\bm{\gamma}_{t}                                  \\
          \nu_{\sigma_{2}}\mid\text{rest} & \sim \text{Inverse-Gamma}\qty(1,\ \dfrac{1}{\sigma^{2}_{2}} + \dfrac{1}{10^{2}}).
        \end{align*}
  \item Regarding $\rho$, because we cannot analytically derive the full conditional distribution, we use the Metropolis algorithm:\begin{enumerate}
          \item Let $\rho^{(c)}$ be the current value. Generate $\rho^{*}$ from $N(\rho^{(c)},k^{2})$.
          \item Using the following, we compute likelihood $p(\rho^{c}\mid\text{rest})$ and $p(\rho^{*}\mid\text{rest})$, respectively:\begin{align*}
                  &p(\rho\mid\text{rest}) \\
                  &= \prod_{t=1}^{T_{0}}|\bm{A}|\exp(-\dfrac{1}{2\sigma^{2}}\big(\bm{A}\bm{Y}_{t}^{c}-\bm{X}_{t}\bm{\beta}-\bm{\eta}^{0}\bm{\gamma}_{t}\big)^{\top}\big(\bm{A}\bm{Y}_{t}^{c}-\bm{X}_{t}\bm{\beta}-\bm{\eta}^{0}\bm{\gamma}_{t}\big)),
                \end{align*}
                where $\bm{A}=\bm{I}_{N}-\rho\bm{W} - \rho\bm{w}\bm{\alpha}^{\top}$.
          \item Let $r=\min\{1,p(\rho^{*}\mid\text{rest})/p(\rho^{(c)}\mid\text{rest})\}$ be the acceptance probability. Subsequently, we generate $u\sim U(0,1)$ and update $\rho^{(c)}$ as $\rho^{*}$ if $r>u$, otherwise stay $\rho^{(c)}$.
        \end{enumerate}
        This algorithm is referred from \citet{lesage2009introduction}.
\end{enumerate}

\section{MCMC Validation and Diagnostics}
\label{app:validation}

\subsection{Geweke (2004) Joint Distribution Test (JDT)}
\label{app:jdt}

To validate the correctness of our MCMC implementation, we conducted the joint distribution test (JDT) proposed by \citet{geweke2004getting}.  
The purpose of the JDT is to verify that an MCMC sampler is free of implementation errors by checking whether it reproduces the correct joint distribution \( p(\theta, y) \).

The test compares two simulators:

\begin{itemize}
    \item \textit{Marginal–Conditional (MC) simulator:}  
    draws parameters independently from the prior distribution and then generates synthetic data directly from the model.
    \item \textit{Successive–Conditional (SC) simulator:}  
    uses the implemented MCMC algorithm to generate draws of both parameters and data sequentially.
\end{itemize}

Let \( h(\theta, y) \) denote any scalar function of the parameters and data.  
If the sampler is correctly implemented, then for each statistic,
\[
Z = 
\frac{\bar{h}_{SC} - \bar{h}_{MC}}
{\sqrt{\mathrm{Var}(h_{SC})/n_{SC} + \mathrm{Var}(h_{MC})/n_{MC}}}
\]
should asymptotically follow the standard normal distribution.

Because our estimation consists of two steps, we applied the JDT separately to (i) the BSCM sampler used to draw the synthetic control weights \( \alpha \), and (ii) the SAR sampler for the autoregressive parameter \( \rho \). As shown in Table~\ref{tab:jdt_summary}, all resulting \( Z \)-statistics lie well within conventional acceptance regions, indicating that our MCMC implementation correctly reproduces the intended joint distribution.

\begin{table}[H]

\caption{Summary of Joint Distribution Test Results}
\label{tab:jdt_summary}
\centering
\begin{tabular}[t]{lrrrrrr}
\toprule
Statistic $g$ & Mean (iid) & Mean (MCMC) & SE (iid) & SE (MCMC) & $Z$ & $p$-value\\
\midrule
$\rho$ & 0.0001542 & -0.0160376 & 0.0003454 & 0.0095749 & 1.6899527 & 0.0910370\\
$\log(\sigma^2)$ & 0.5758153 & 0.5812187 & 0.0009061 & 0.0147590 & -0.3654197 & 0.7147982\\
$\bar{y}_c$ & -0.0000785 & 0.0074088 & 0.0007557 & 0.0062590 & -1.1876386 & 0.2349758\\
$\log(\mathrm{Var}(y_c))$ & 2.3945838 & 2.4063212 & 0.0008891 & 0.0169414 & -0.6918719 & 0.4890178\\
$h_{\text{spatial}}$ & 10.0261618 & 5.8165006 & 4.7311009 & 10.3394342 & 0.3702281 & 0.7112126\\
\addlinespace
$\mathrm{Corr}(y_0, Wy)$ & 0.0004321 & -0.0004719 & 0.0001928 & 0.0014144 & 0.6332193 & 0.5265905\\
$\bar{\beta}$ & -0.0005247 & 0.0034374 & 0.0005002 & 0.0060446 & -0.6532536 & 0.5135928\\
$\bar{\eta}$ & -0.0001587 & -0.0017863 & 0.0002500 & 0.0012491 & 1.2776321 & 0.2013792\\
$\bar{\Gamma}$ & -0.0002017 & -0.0009302 & 0.0001824 & 0.0006937 & 1.0156151 & 0.3098127\\
\bottomrule
\end{tabular}
\begin{tablenotes}
  \footnotesize
  \item \textit{Notes:} This table reports the results of Geweke's (2004) joint distribution test, comparing the moments generated by the marginal-conditional (MC) simulator and the successive-conditional (SC) simulator. The $Z$-statistic is expected to follow a standard normal distribution under the null hypothesis that the MCMC implementation is correct. $|Z| < 1.96$ indicates no significant difference at the 5\% level.
\end{tablenotes}
\end{table}

\subsection{MCMC Convergence Diagnostics}
\label{app:convergence}

We evaluated MCMC convergence using standard diagnostic tools, including trace plots, posterior density plots, and effective sample size (ESS) statistics for the key parameters \( \rho \) and \( \sigma^{2} \). The diagnostics uniformly indicate good mixing behavior and adequate effective sample sizes.

\subsubsection{California Tobacco Tax Application}

Figure~\ref{fig:ca_diag} presents the trace plots and posterior densities for the California application. The chains exhibit stable behavior and no signs of non-convergence.  
Table~\ref{tab:ca_diag} reports summary statistics and ESS values, all of which exceed commonly used thresholds.

\bigskip

\begin{figure}[h!]
\caption{MCMC Convergence Diagnostics (California)}
    \centering
    \includegraphics[width=1.0\textwidth]{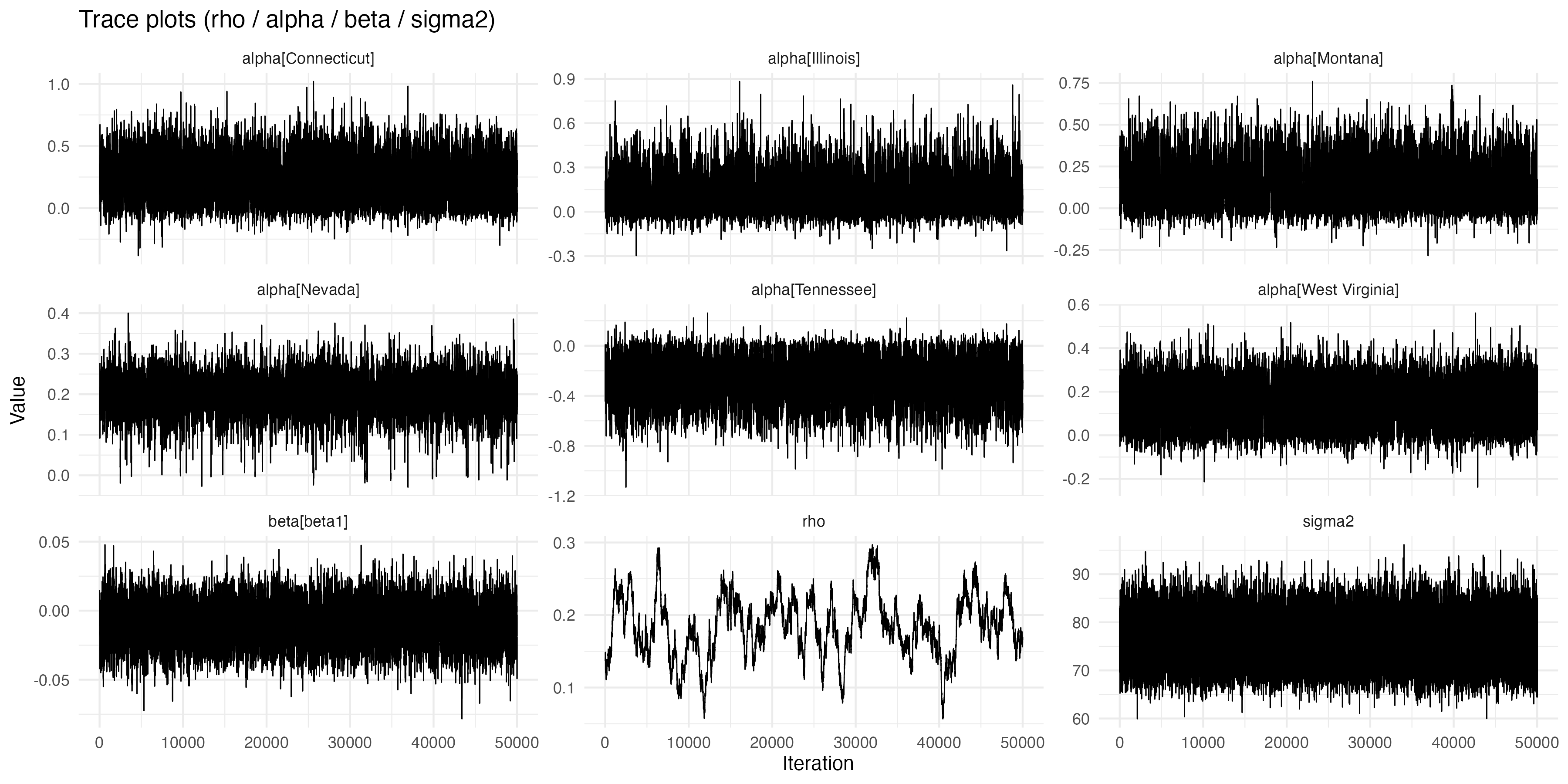}
    \label{fig:ca_diag}
    \begin{tablenotes}
  \footnotesize
  \item \textit{Notes:} The figure displays trace plots for representative parameters---$\rho$ (labeled as ``rho''), $\sigma^2$ (labeled as ``sigma2''), and selected elements of $\bm{\alpha}$ (labeled as ``alpha'') and $\bm{\beta}$ (labeled as ``beta'')---after discarding burn-in samples. The plots demonstrate stable mixing and stationarity of the Markov chains, with no evidence of divergent behavior.
\end{tablenotes}
\end{figure}

\begin{table}[H]
\centering
\caption{\label{tab:tab:ca_diag}MCMC Posterior Summary and Diagnostics (California)}
\label{tab:ca_diag}
\centering
\footnotesize
\setlength{\tabcolsep}{2pt} 
\begin{tabular}[t]{lrrrrrrr}
\toprule
Parameter & Mean & SD & $2.5\%$ & $50\%$ & $97.5\%$ & ESS & $\hat{R}$\\
\midrule
rho & 0.1847 & 0.0420 & 0.0957 & 0.1848 & 0.2677 & 43.9816 & 1.001\\
alpha[Tennessee] & -0.2580 & 0.1662 & -0.5880 & -0.2618 & 0.0091 & 1515.7641 & 1.000\\
alpha[Connecticut] & 0.2176 & 0.1654 & -0.0308 & 0.2173 & 0.5435 & 1620.1672 & 1.000\\
alpha[Nevada] & 0.1997 & 0.0366 & 0.1204 & 0.2015 & 0.2681 & 3856.4972 & 1.000\\
alpha[West Virginia] & 0.1321 & 0.0950 & -0.0178 & 0.1361 & 0.3116 & 1971.8119 & 1.001\\
\addlinespace
alpha[Montana] & 0.1183 & 0.1263 & -0.0307 & 0.0836 & 0.3993 & 1550.3598 & 1.000\\
alpha[Illinois] & 0.1115 & 0.1168 & -0.0326 & 0.0856 & 0.3893 & 1889.7454 & 1.001\\
sigma2 & 76.0961 & 4.3095 & 68.1694 & 75.9431 & 84.9599 & 7415.0698 & 1.000\\
beta[beta1] & -0.0097 & 0.0140 & -0.0370 & -0.0097 & 0.0180 & 2267.8469 & 1.000\\
\bottomrule
\end{tabular}
\begin{tablenotes}
  \footnotesize
  \item \textit{Notes:} This table summarizes the posterior distributions of key parameters. ``ESS" denotes the Effective Sample Size, which estimates the number of independent samples. "$R$" (Split $\hat{R}$) is the Gelman-Rubin convergence diagnostic; values close to 1.0 (typically $< 1.01$) indicate convergence.
\end{tablenotes}
\end{table}

\subsubsection{Sudan Split Application}

A similar inspection for the Sudan Split analysis (Figure~\ref{fig:sudan_diag} and Table~\ref{tab:sudan_diag}) confirms excellent convergence and well-mixed Markov chains.

\bigskip

\begin{figure}[h!]
\caption{MCMC Convergence Diagnostics (Sudan Split)}
    \centering
    \includegraphics[width=1.0\textwidth]{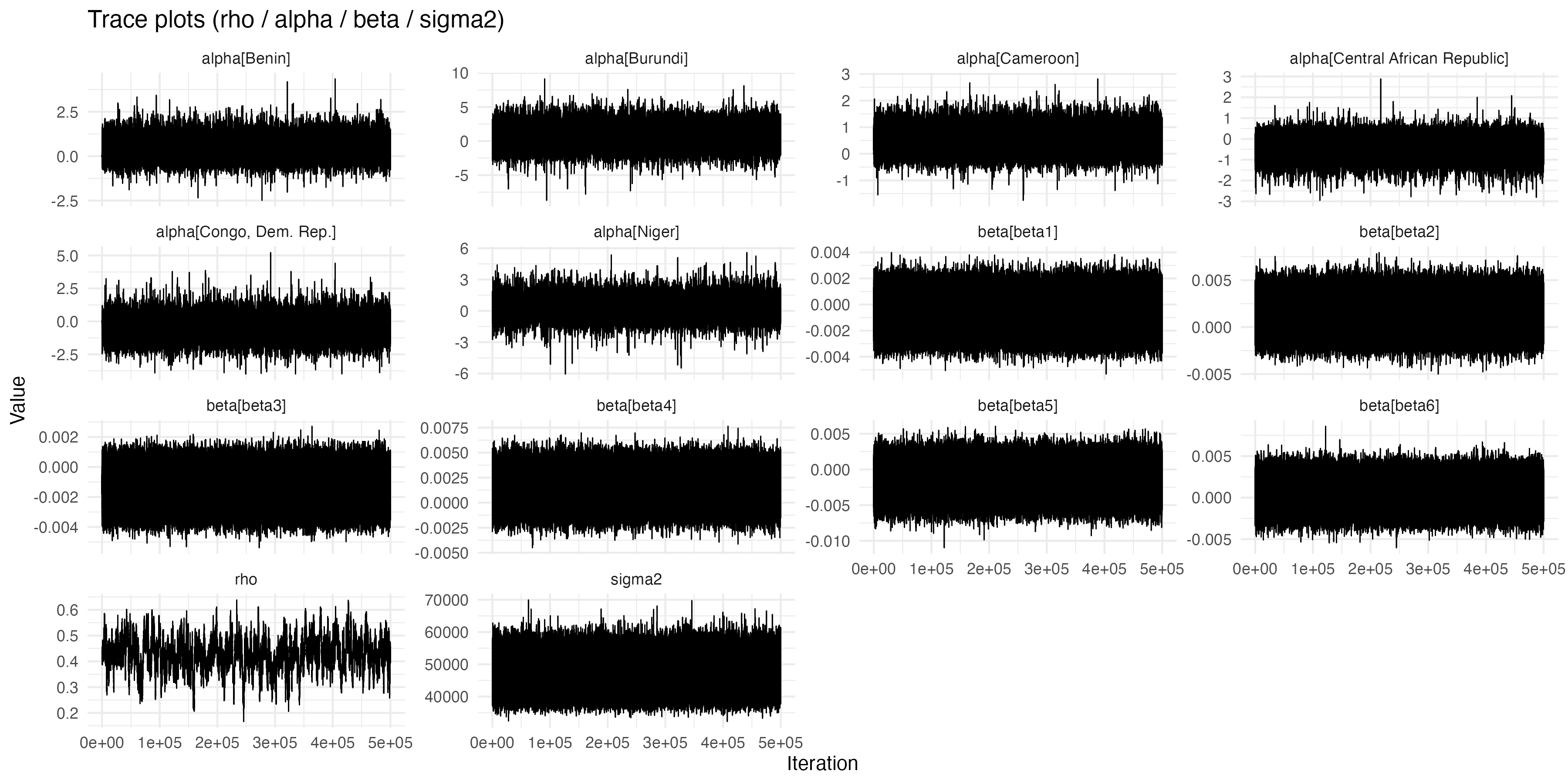}
    \label{fig:sudan_diag}
    \begin{tablenotes}
  \small
  \item \textit{Notes:} The figure displays trace plots for representative parameters---$\rho$ (labeled as ``rho''), $\sigma^2$ (labeled as ``sigma2''), and selected elements of $\bm{\alpha}$ (labeled as ``alpha'') and $\bm{\beta}$ (labeled as ``beta'')---after discarding burn-in samples. The plots demonstrate stable mixing and stationarity of the Markov chains, with no evidence of divergent behavior.
\end{tablenotes}
\end{figure}

\bigskip

\begin{table}[!h]
\centering
\caption{MCMC Posterior Summary and Diagnostics (Sudan Split)}
\label{tab:sudan_diag}
\centering
\resizebox{\ifdim\width>\linewidth\linewidth\else\width\fi}{!}{
\begin{tabular}[t]{lrrrrrrr}
\toprule
Parameter & Mean & SD & $2.5\%$ & $50\%$ & $97.5\%$ & ESS & $\hat{R}$\\
\midrule
rho & 0.4274 & 0.0642 & 0.2941 & 0.4295 & 0.5498 & 387.7632 & 1\\
alpha[Burundi] & 0.3084 & 0.8360 & -1.0837 & 0.0897 & 2.4226 & 41314.1005 & 1\\
alpha[Cameroon] & 0.2646 & 0.3380 & -0.1046 & 0.1183 & 1.0550 & 18294.9308 & 1\\
alpha[Benin] & 0.1913 & 0.3343 & -0.2246 & 0.0678 & 1.0814 & 39594.9780 & 1\\
alpha[Congo, Dem. Rep.] & -0.1668 & 0.4114 & -1.2907 & -0.0326 & 0.3997 & 37242.3548 & 1\\
\addlinespace
alpha[Central African Republic] & -0.1216 & 0.2567 & -0.8185 & -0.0369 & 0.2127 & 53376.2961 & 1\\
alpha[Niger] & 0.1184 & 0.4191 & -0.5784 & 0.0284 & 1.1822 & 84888.2934 & 1\\
sigma2 & 46351.6902 & 3717.6094 & 39616.2464 & 46153.3551 & 54205.1948 & 174641.3582 & 1\\
beta[beta1] & -0.0006 & 0.0010 & -0.0025 & -0.0006 & 0.0014 & 160118.0630 & 1\\
beta[beta2] & 0.0014 & 0.0014 & -0.0012 & 0.0014 & 0.0041 & 211738.0450 & 1\\
\addlinespace
beta[beta3] & -0.0014 & 0.0009 & -0.0031 & -0.0014 & 0.0003 & 42583.5868 & 1\\
beta[beta4] & 0.0015 & 0.0012 & -0.0009 & 0.0015 & 0.0039 & 16820.7741 & 1\\
beta[beta5] & -0.0015 & 0.0017 & -0.0048 & -0.0015 & 0.0018 & 24412.7901 & 1\\
beta[beta6] & 0.0005 & 0.0013 & -0.0021 & 0.0005 & 0.0032 & 39437.0112 & 1\\
\bottomrule
\end{tabular}}
\begin{tablenotes}
  \footnotesize
  \item \textit{Notes:} This table summarizes the posterior distributions of key parameters. ``ESS" denotes the Effective Sample Size, which estimates the number of independent samples. ``$R$" (Split $\hat{R}$) is the Gelman-Rubin convergence diagnostic; values close to 1.0 (typically $< 1.01$) indicate convergence.
\end{tablenotes}
\end{table}

\bigskip

\section{Model Specification Diagnostics}
\label{app:specification}

\subsection{Prior Predictive Analysis}
\label{app:ppa}

To evaluate whether the model and prior distributions place reasonable mass over data-generating processes consistent with the observed data, we performed a prior predictive analysis following \citet{geweke2005contemporary}.  
We generated \( 100,000 \) synthetic datasets \( Y^{c}_{sim} \) from the prior predictive distribution associated with the Step~2 SAR model.

For each simulated dataset, we computed a set of nine informative summary statistics: the sample mean, the logarithm of the variance, the spatial quadratic form, the correlation between the treated unit and the synthetic control predictor, the first- and second-order autocorrelations, the proportion of variance explained by the first principal component, the average skewness, and the average kurtosis.

Figures~\ref{fig:ca_ppa_hist} and \ref{fig:sudan_ppa_hist} show the prior predictive distributions (histograms) for each statistic, along with the value computed from observed data \( h(y^{o}) \), for the California Tobacco Tax application and the Sudan Split application, respectively.  
For each statistic, we also compute the Bayesian prior predictive \(p\)-value,
\[
P\bigl(h(Y^{c}_{sim}) \le h(y^{o}) \,\big|\, A \bigr),
\]
and report the results in Tables~\ref{tab:ca_ppa} and \ref{tab:sudan_ppa}. 
Overall, the observed statistics fall well within the support of the prior predictive distributions, suggesting that the model assumptions and priors are compatible with the observed data.


\bigskip

\begin{figure}[h!]
 \caption{Prior Predictive Analysis (California Tobacco Tax)}
    \centering
    \includegraphics[width=1.0\textwidth]{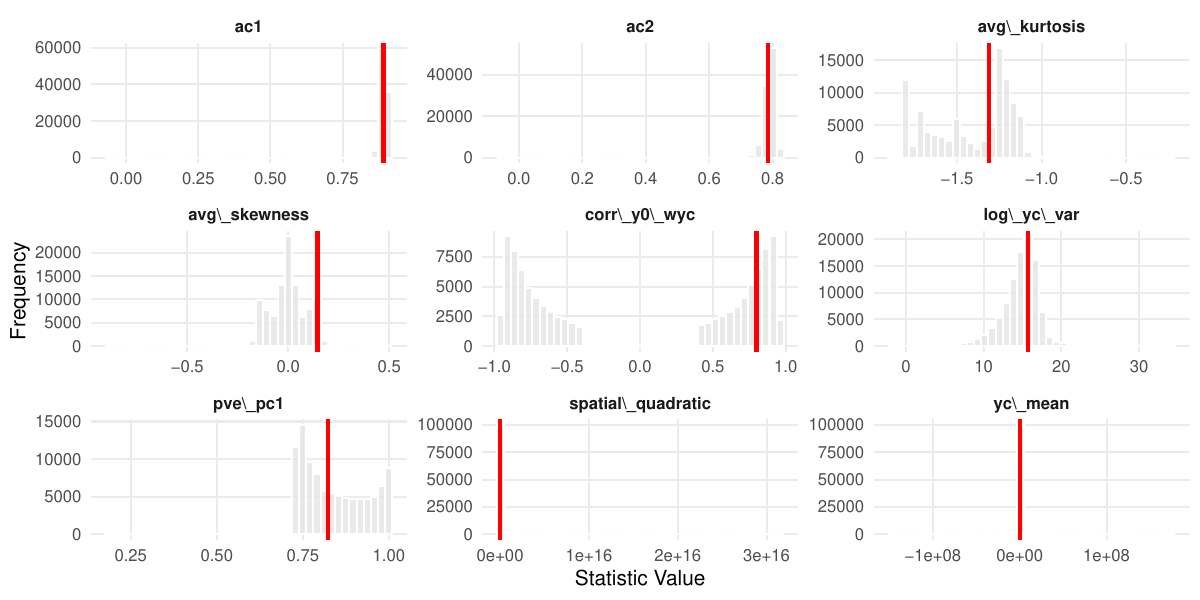}
    \label{fig:ca_ppa_hist}
    \begin{tablenotes}
  \footnotesize
  \item \textit{Notes:} Histograms show the prior predictive distributions of nine summary statistics generated from the model's prior. The vertical red lines indicate the values of the statistics calculated from the actual observed pre-treatment data.
\end{tablenotes}
\end{figure}

\begin{figure}[h!]
\caption{Prior Predictive Analysis (Sudan Split)}
    \centering
    \includegraphics[width=1.0\textwidth]{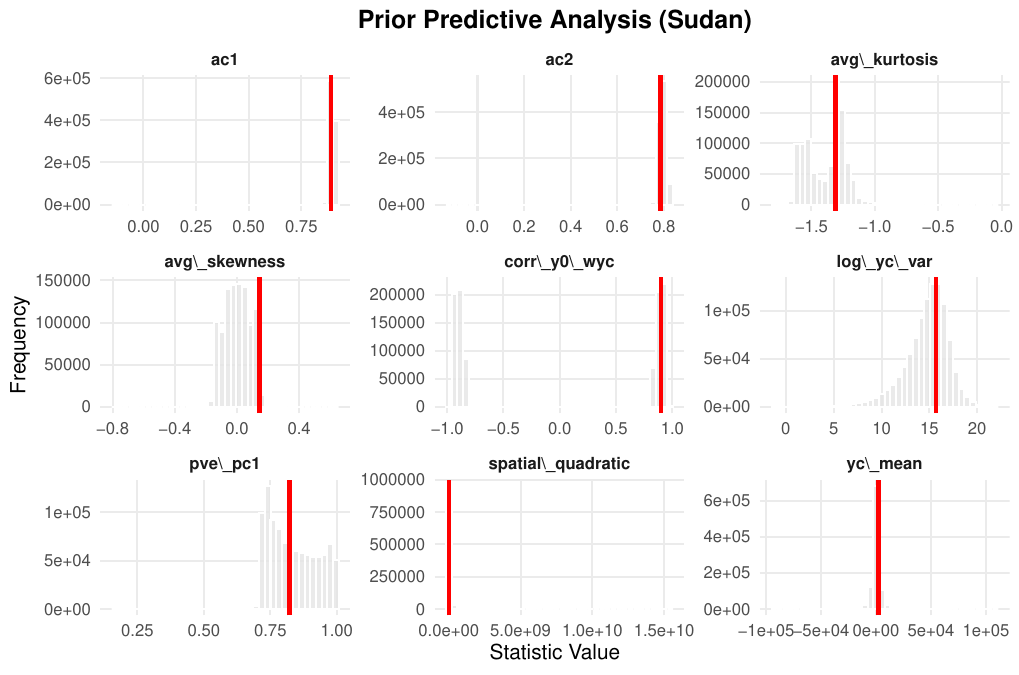}
    \label{fig:sudan_ppa_hist}
    \begin{tablenotes}
  \footnotesize
  \item \textit{Notes:} Histograms show the prior predictive distributions of nine summary statistics generated from the model's prior. The vertical red lines indicate the values of the statistics calculated from the actual observed pre-treatment data.
\end{tablenotes}
\end{figure}

\bigskip

\begin{table}[H]
\centering
\caption{Prior Predictive Analysis Summary Table (California Tobacco Tax)}
\label{tab:ca_ppa}
\centering
\begin{tabular}[t]{lrr}
\toprule
Statistic & Observed & $P(h \leq h(y^{o}) | A)$\\
\midrule
Mean & 131.500 & 0.906\\
log-Var & 6.984 & 0.681\\
Spatial Quad. & 17844.720 & 0.802\\
Corr($Y_0$, $w'Y$) & 0.945 & 0.631\\
AC(1) & 0.894 & 0.976\\
\addlinespace
AC(2) & 0.810 & 0.980\\
PVE(PC1) & 0.625 & 0.049\\
Skewness & -0.396 & 0.480\\
Kurtosis & -0.943 & 0.715\\
\bottomrule
\end{tabular}
\begin{tablenotes}
  \footnotesize
  \item \textit{Notes:} The table reports the Bayesian p-values for the prior predictive check, defined as the proportion of simulated statistics that are less than or equal to the observed statistic, $P(h(y_{sim}) \le h(y_{obs}) \mid \text{Prior})$. Values far from 0 or 1 indicate consistency between the prior/model and the data.
\end{tablenotes}
\end{table}


\bigskip

\begin{table}[!h]
\centering
\caption{Prior Predictive Analysis Summary Table (Sudan Split)}
\label{tab:sudan_ppa}
\centering
\begin{tabular}[t]{lrr}
\toprule
Statistic & Observed & $P(h \leq h(y^{o}) | A)$\\
\midrule
Mean & 2142.779 & 0.806\\
log-Var & 15.666 & 0.625\\
Spatial Quad. & 7464968.100 & 0.614\\
Corr($Y_0$, $w'Y$) & 0.799 & 0.750\\
AC(1) & 0.892 & 0.339\\
\addlinespace
AC(2) & 0.786 & 0.294\\
PVE(PC1) & 0.823 & 0.503\\
Skewness & 0.145 & 0.923\\
Kurtosis & -1.312 & 0.501\\
\bottomrule
\end{tabular}
\begin{tablenotes}
  \footnotesize
  \item \textit{Notes:} The table reports the Bayesian p-values for the prior predictive check, defined as the proportion of simulated statistics that are less than or equal to the observed statistic, $P(h(y_{sim}) \le h(y_{obs}) \mid \text{Prior})$. Values far from 0 or 1 indicate consistency between the prior/model and the data.
\end{tablenotes}
\end{table}

\subsection{Prior Sensitivity Analysis}
\label{app:sensitivity}

We investigated the sensitivity of posterior inferences to the choice of prior hyperparameters, with a focus on the priors for \( \rho \) and \( \sigma^{2} \).  
For several values of the hyperparameters \( (a_{0}, b_{0}, \rho_{\mathrm{lo}}, \rho_{\mathrm{hi}}) \), we recomputed the posterior distribution of all model parameters.

We focus on the parameters \(\rho\), \(\sigma^{2}\), and \(\beta_{1}\).
For each prior specification, we report the posterior mean and standard
deviation, together with the 2.5th and 97.5th percentiles of the posterior
distribution. The hyperparameters \((a_{0}, b_{0})\) correspond to the prior placed on \(\sigma^{2}\). 

Tables~\ref{tab:ca_sens} and \ref{tab:sudan_sens} report these quantities for the California Tobacco Tax application and the Sudan Split application, respectively. 
The posterior mean of \( \rho \) remains stable across the different prior choices, indicating that the results are not sensitive to the prior specification.


\begin{table}[H]
\centering
\caption{Prior Sensitivity Analysis (California Tobacco Tax)}
\label{tab:ca_sens}
\centering
\resizebox{\ifdim\width>\linewidth\linewidth\else\width\fi}{!}{
\begin{tabular}[t]{rrrrrlrrrr}
\toprule
$a_0$ & $b_0$ & $\rho_{\mathrm{lo}}$ & $\rho_{\mathrm{hi}}$ & $\Delta\rho$ & Parameter & Mean & SD & $2.5\%$ & $97.5\%$\\
\midrule
3 & 1.0 & -0.5 & 0.5 & 0.05 & rho & 0.4968491 & 0.0031053 & 0.4884527 & 0.4999165\\
3 & 1.0 & -0.5 & 0.5 & 0.05 & sigma2 & 1774.4801797 & 96.1709457 & 1596.2061721 & 1972.7777152\\
3 & 1.0 & -0.5 & 0.5 & 0.05 & beta[1] & 0.8648020 & 0.0250859 & 0.8159419 & 0.9143218\\
5 & 1.0 & -0.3 & 0.3 & 0.03 & rho & 0.2980091 & 0.0019977 & 0.2926308 & 0.2999507\\
5 & 1.0 & -0.3 & 0.3 & 0.03 & sigma2 & 2390.5653694 & 129.0438905 & 2150.5123794 & 2656.2601841\\
\addlinespace
5 & 1.0 & -0.3 & 0.3 & 0.03 & beta[1] & 1.2450636 & 0.0285566 & 1.1891587 & 1.3010971\\
2 & 0.5 & -0.7 & 0.7 & 0.07 & rho & 0.6539277 & 0.0192914 & 0.6147986 & 0.6900835\\
2 & 0.5 & -0.7 & 0.7 & 0.07 & sigma2 & 1511.6332085 & 84.4858128 & 1355.1972011 & 1686.1148914\\
2 & 0.5 & -0.7 & 0.7 & 0.07 & beta[1] & 0.5642521 & 0.0432292 & 0.4826236 & 0.6511700\\
\bottomrule
\end{tabular}}
\begin{tablenotes}
  \footnotesize
  \item \textit{Notes:} This table presents the posterior means, standard deviations, and 95\% credible intervals for key parameters ($\rho$, $\sigma^2$, $\beta_1$) under alternative prior hyperparameter settings. $(a_0, b_0)$ correspond to the Inverse-Gamma prior for $\sigma^2$, and $(\rho_{lo}, \rho_{hi})$ define the Uniform prior support for $\rho$. The stability of the estimates across different settings indicates robustness to prior specification.
\end{tablenotes}
\end{table}


\begin{table}[!h]
\centering
\caption{Prior Sensitivity Analysis (Sudan Split)}
\label{tab:sudan_sens}
\centering
\resizebox{\ifdim\width>\linewidth\linewidth\else\width\fi}{!}{
\begin{tabular}[t]{rrrrrlrrrr}
\toprule
a0 & b0 & rho\_lo & rho\_hi & step\_rho & param & mean & sd & q025 & q975\\
\midrule
3 & 1.0 & -0.5 & 0.5 & 0.05 & rho & 0.0000003 & 0.0000017 & -0.0000024 & 0.0000025\\
3 & 1.0 & -0.5 & 0.5 & 0.05 & sigma2 & 0.0054748 & 0.0004059 & 0.0047372 & 0.0063262\\
3 & 1.0 & -0.5 & 0.5 & 0.05 & beta[1] & 0.9999998 & 0.0000016 & 0.9999967 & 1.0000030\\
5 & 1.0 & -0.3 & 0.3 & 0.03 & rho & -0.0000003 & 0.0000014 & -0.0000025 & 0.0000030\\
5 & 1.0 & -0.3 & 0.3 & 0.03 & sigma2 & 0.0054130 & 0.0003992 & 0.0046855 & 0.0062497\\
\addlinespace
5 & 1.0 & -0.3 & 0.3 & 0.03 & beta[1] & 1.0000002 & 0.0000015 & 0.9999971 & 1.0000031\\
2 & 0.5 & -0.7 & 0.7 & 0.07 & rho & -0.0000004 & 0.0000015 & -0.0000019 & 0.0000012\\
2 & 0.5 & -0.7 & 0.7 & 0.07 & sigma2 & 0.0027566 & 0.0002052 & 0.0023831 & 0.0031867\\
2 & 0.5 & -0.7 & 0.7 & 0.07 & beta[1] & 1.0000003 & 0.0000013 & 0.9999979 & 1.0000028\\
\bottomrule
\end{tabular}}
\begin{tablenotes}
  \footnotesize
  \item \textit{Notes:} This table presents the posterior means, standard deviations, and 95\% credible intervals for key parameters ($\rho$, $\sigma^2$, $\beta_1$) under alternative prior hyperparameter settings. $(a_0, b_0)$ correspond to the Inverse-Gamma prior for $\sigma^2$, and $(\rho_{lo}, \rho_{hi})$ define the Uniform prior support for $\rho$. The stability of the estimates across different settings indicates robustness to prior specification.
\end{tablenotes}
\end{table}


\end{document}